\documentclass[letterpaper]{article} 
\usepackage{aaai24}  
\usepackage{times}  
\usepackage{helvet}  
\usepackage{courier}  
\usepackage[hyphens]{url}  
\usepackage{graphicx} 
\usepackage{natbib}  
\usepackage{caption} 
\usepackage{algorithm}
\usepackage{algorithmic}

\usepackage{amssymb}
\usepackage{amsmath}
\usepackage{amsthm}
\usepackage{multirow}
\usepackage{thmtools,thm-restate}
\newtheorem{theorem}{Theorem}
\newtheorem{corollary}{Corollary}[theorem]

\title{Backpropagation Through Agents}
\author{
    Zhiyuan Li\textsuperscript{\rm 1},
    Wenshuai Zhao\textsuperscript{\rm 2},
    Lijun Wu\textsuperscript{\rm 1},
   Joni Pajarinen\textsuperscript{\rm 2}
}
\affiliations{
    \textsuperscript{\rm 1}School of Computer Science and Engineering, University of Electronic Science and Technology of China\\
     \textsuperscript{\rm 2}Department of Electrical Engineering and Automation, Aalto University\\

    zhiyuanli@std.uestc.edu.cn, \{wenshuai.zhao,joni.pajarinen\}@aalto.fi, wljuestc@sina.com 
}

\begin{document}

\maketitle

\begin{abstract}
A fundamental challenge in multi-agent reinforcement learning (MARL) is to learn the joint policy in an extremely large search space, which grows exponentially with the number of agents. Moreover, fully decentralized policy factorization significantly restricts the search space, which may lead to sub-optimal policies. In contrast, the auto-regressive joint policy can represent a much richer class of joint policies by factorizing the joint policy into the product of a series of conditional individual policies. While such factorization introduces the action dependency among agents explicitly in sequential execution, it does not take full advantage of the dependency during learning. In particular, the subsequent agents do not give the preceding agents feedback about their decisions. In this paper, we propose a new framework Back-Propagation Through Agents (BPTA) that directly accounts for both agents' own policy updates and the learning of their dependent counterparts. This is achieved by propagating the feedback through action chains. With the proposed framework, our Bidirectional Proximal Policy Optimisation (BPPO) outperforms the state-of-the-art methods. Extensive experiments on matrix games, StarCraftII v2, Multi-agent MuJoCo, and Google Research Football demonstrate the effectiveness of the proposed method.
\end{abstract}

\section{Introduction}

Multi-agent reinforcement learning (MARL) is a promising approach to many real-world applications that naturally comprise multiple decision-makers interacting at the same time, such as cooperative robotics \cite{10.5555/3545946.3598752}, traffic management \cite{DBLP:conf/atal/MaW20}, and autonomous driving \cite{DBLP:journals/corr/Shalev-ShwartzS16a}. Although reinforcement learning (RL) has recorded sublime success in various single-agent domains, trivially applying single-agent RL algorithms in this setting brings about the curse of dimensionality. In multi-agent settings, agents need to explore an extremely large policy space, which grows exponentially with the team size, to learn the optimal joint policy.

Existing popular multi-agent policy gradient (MAPG) frameworks \cite{10.5555/3295222.3295385, 10.5555/3504035.3504398, NEURIPS2022_9c1535a0, wang2023order, pmlr-v139-zhang21m} often directly represent the joint policy as the Cartesian Product of each agent’s fully independent policy. However, this factorization ignores the coordination between agents and severely limits the complexity of the joint policy, causing the learning algorithm to converge to a \textit{Pareto-dominant equilibrium} \cite{christianos2022pareto}. This phenomenon is commonly referred to as \textit{relative overgeneralization} \cite{DBLP:conf/aaaiss/WeiWFL18, wang2021dop} and can occur even in simple scenarios \cite{pmlr-v162-fu22d, ye2023global}. To tackle these issues, some recent works \cite{wang2023more, NEURIPS2022_69413f87, pmlr-v162-fu22d} present the joint policy in an \textit{auto-regressive} form based on the chain rule \cite{box2015time}. The auto-regressive model specifies that an agent's policy depends on its preceding agents' actions. In this way, the dependency among agents’ policies is explicitly considered and the expressive limitations of the joint policy can be significantly relaxed. However, they only take into account the preceding agents‘ actions during decision-making, i.e., the forward process, while disregarding reactions from subsequent agents during policy improvement, i.e., the backward process \cite{li2023ace}. This may lead to conflicting directions in policy updates for individual agents, where their local improvements may jointly result in worse outcomes. In contrast, the neural circuits in the central nervous system responsible for the sensorimotor loop consist of two internal models \cite{MULLER2021661}: 1) the forward model, which builds the causal flow by integrating the joint actions, and 2) the backward model, which maps the relation between an action and its consequence to invert the causal flow. Such two bidirectional models internally interact in order to enhance learning mechanisms.

In this paper, we aim to augment the existing MAPG framework with bidirectional dependency \cite{li2023ace}, i.e. forward and backward process, to provide richer peer feedback and align the policy improvement directions of individual agents with that of the joint policy. To this end, we propose \textit{Back-Propagation Through Agents} (BPTA), a multi-agent reinforcement learning framework that follows the Back-Propagation Through Time (BPTT) used for training recursive neural networks (RNN) \cite{DBLP:conf/emnlp/ChoMGBBSB14}. Specifically, BPTA begins by unfolding the execution sequence in agents. The actions passed to subsequent agents during the forward process will be integrated with their own actions and serve as latent variables \cite{kingma2022autoencoding} in the backward process, the reactions from the subsequent agents are propagated to the preceding agents through the variables using the reparameterization trick \cite{NIPS2015_bc731692}. Taking the feedback from subsequent agents into account allows each agent to learn from the consequences of the collective actions and adapt to the changing behavior of the team. Furthermore, dependent on such rich feedback agents can complete the causality loop: cyclic interaction between the forward and backward process. As a result, BPTA enables individuals to function as a whole and find a consistent improvement direction. We incorporate PPO with auto-regressive policy and BPTA and propose Bidirectional Proximal Policy Optimisation (BPPO). Empirically, in several tasks, including matrix game \cite{10.5555/295240.295800}, Google Research Football (GRF) \cite{Kurach_Raichuk_Stańczyk_Zając_Bachem_Espeholt_Riquelme_Vincent_Michalski_Bousquet_Gelly_2020}, StarCraft Multi-Agent Challenge Version 2 (SMACv2) \cite{ellis2022smacv2}, and Multi-agent MuJoCo (MA-MuJoCo), BPPO achieves better performance than baselines. 

Specifically, our contribution is summarized as follows.

\begin{itemize}
    \item We propose a novel framework BPTA that, for the first time, explicitly models feedback from action-dependent peer agents. In particular, BPTA allows derivatives to pass across agents during learning.
    \item Our proposed framework can be naturally integrated with existing conditional policy-gradient methods. We augment PPO with the auto-regressive policy under the BPTA framework and propose Bidirectional Proximal Policy Optimisation (BPPO).
    \item Finally, the effectiveness of the proposed method is verified in four cooperative environments, and empirical results show that the proposed method outperforms the state-of-the-art algorithms.
\end{itemize}

\section{Related Work}

Various works on MARL have been proposed to tackle cooperative tasks, including algorithms in which agents make decisions simultaneously and algorithms that coordinate agents' actions based on static or dynamic execution orders. 

\textbf{Simultaneous decision scheme.} Most algorithms tend to follow a simultaneous decision scheme, where agents' policies are only conditioned on their individual observations. One line of research extends PG from RL to MARL \cite{10.5555/3295222.3295385, 10.5555/3504035.3504398, wang2021dop, NEURIPS2022_9c1535a0, wang2023order, pmlr-v139-zhang21m}, adopting the Actor-Critic (AC) approach, where each actor explicitly represents the independent policy, and the estimated centralized value function is known as the critic. Under this scheme, in contrast to independent updates, some recent methods sequentially execute agent-by-agent updates, such as Rollout and Policy Iteration for a Single Agent (RPISA) \cite{9317713}, Heterogeneous PPO (HAPPO) \cite{kuba2022trust}, and Agent-by-agent Policy Optimization (A2PO) \cite{wang2023order}. Another line is value-based methods, where the joint Q-function is decomposed into individual utility functions following different interpretations of Individual-Global-Maximum (IGM) \cite{10.5555/3237383.3238080, 10.5555/3455716.3455894, pmlr-v97-son19a, wang2021qplex, pmlr-v162-wan22c}. VDN \cite{10.5555/3237383.3238080} and QMIX \cite{10.5555/3455716.3455894} provide sufficient conditions for IGM, however, suffer from structural constraints. QTRAN \cite{pmlr-v97-son19a} and QPLEX \cite{wang2021qplex} complete the representation capacity of the joint Q-function through optimization constraints and a dueling mixing network respectively, while it is impractical in complicated tasks. \citeauthor{pmlr-v162-wan22c} introduce Optimal consistency and True-Global-Max (TGM), then propose GVR to ensure the optimality. A special case is SeCA \cite{10.5555/3545946.3598674}, which factorizes the joint policy evaluation into a sequence of successive evaluations.

\textbf{Sequential decision scheme.} In this scheme, algorithms explicitly model the coordination among agents via actions. One perspective is the auto-regressive paradigm, where agents make decisions sequentially \cite{NEURIPS2022_69413f87, pmlr-v162-fu22d, ye2023global, wang2023more, li2023ace}. MAT \cite{NEURIPS2022_69413f87} transform MARL into a sequence modeling problem, and introduce \textit{Transformer} \cite{NIPS2017_3f5ee243} to generate solutions. However, MAT may fail to achieve the monotonic improvement guarantee as it does not follow the sequential update scheme. \citeauthor{wang2023more} derives the multi-agent conditional factorized soft policy iteration theorem by incorporating auto-regressive policy into SAC \cite{haarnoja2019soft}. ACE \cite{li2023ace} and TAD \cite{ye2023global} first cast the Multi-agent Markov decision process (MMDP) \cite{10.5555/3091574.3091594} as an equivalent single-agent Markov decision process (MDP), and solve the single-agent MDP with Q-learning and PPO, respectively. However, only ACE considers the reactions from subsequent agents by calculating the maximum Q-value over the possible actions of the successors. In another perspective, the interactions between agents are modeled by a coordination graph \cite{10.5555/3535850.3535976, pmlr-v162-yang22a}. However, the updates of the agents in the graph are independent of the subsequent agents.

In contrast, our proposed BPTA augmented auto-regressive method lies in the second category and is the first bidirectional PG-based MARL method. 

\section{Background}
\subsection{Problem Formulation}
In MARL, a \textit{decentralized partially observable Markov decision process} (Dec-POMDP) \cite{10.5555/2967142} is commonly applied to model the interaction among agents within a shared environment under partial observability. A Dec-POMDP is defined by a tuple $G = \left<\mathcal{N},\mathcal{S},\mathcal{A}, \mathcal{P},\Omega, O,\mathcal{R},\rho_0, \gamma \right>$, where $\mathcal{N} = \left\{ 1,\ldots,n \right\}$ is a set of agents, $s \in \mathcal{S}$ denotes the state of the environment, $\mathcal{A} = \prod_{i=1}^{n}A^i$ is the product of the agents' action spaces, namely the joint action space, $\Omega=\prod_{i=1}^{n}\Omega^i$ is the set of joint observations, and $\rho_0$ is the distribution of the initial state. At time step $t \in \mathbb{N}$, each agent $i \in \mathcal{N}$ takes an action $a^i_t$ according to its policy $\pi ^i(\cdot | o_t^i)$, where $o_t^i$ is drawn from the observation function $O(s_t, i)$. With the joint observation $\textbf{o}_t = \left\{ o^1_t, \ldots,o^n_t \right\}$ and the joint action of agents $\textbf{a}_t = \left\{ a^1_t, \ldots,a^n_t \right\}$ drawn from the joint policy $\boldsymbol\pi \left ( \cdot | \textbf{o}_t \right )$, the environment moves to a state $s'$ with probability $\mathcal{P}\left ( s' | s, a_t \right )$, and each agent receives a joint reward $r_t = \mathcal{R}\left ( s_t, a_t \right ) \in \mathbb{R}$. The state value function, the state-action value function, and the advantage function are defined as $V_{\boldsymbol\pi}(s) \triangleq \mathbb{E}_{\textbf{a}_{0:\infty \sim \boldsymbol\pi} , s_{1:\infty \sim \mathcal{P}}}[\Sigma_{t=0}^\infty\gamma ^{t}r_{t}|s_{0}=s]$, $Q_{\boldsymbol\pi}(s,\textbf{a}) \triangleq \mathbb{E}_{\textbf{a}_{1:\infty \sim \boldsymbol\pi} , s_{1:\infty \sim \mathcal{P}}}[\Sigma_{t=0}^\infty\gamma ^{t}r_{t}|s_{0}=s, \textbf{a}_0=\textbf{a}]$, and $A_{\boldsymbol\pi}(s,\textbf{a}) \triangleq Q_{\boldsymbol\pi}(s,\textbf{a}) - V_{\boldsymbol\pi}(s) $. The agents aim to maximize the expected total reward: 
\begin{equation}
\label{objective}
\mathcal{J}(\boldsymbol\pi)\triangleq \mathbb{E}_{s_{0},\textbf{a}_0,\dots }\left [ \sum_{t=0}^{\infty}\gamma ^{t}r_t \right ],
\end{equation}
where $s_0 \sim \rho _0 (s_0), \textbf{a}_t \sim \boldsymbol\pi (\textbf{a}_t|s_t)$. In order to keep the notation concise, we will use state $s$ in the subsequent equations.

\subsection{Independent Multi-Agent Stochastic Policy Gradient}
In cooperative MARL tasks, popular PG methods follow fully independent factorization: $\boldsymbol\pi(\textbf{a}|s)=\prod_{i=1}^{n}\pi_{\theta _{i}}(a^i|s)$. With such a form, following along the standard Stochastic Policy Gradient Theorem, \citeauthor{DBLP:conf/aaaiss/WeiWFL18} derive the independent multi-agent policy gradient estimator for the cooperative MARL:
\begin{equation}
\begin{aligned}
\bigtriangledown _{\theta _i}\mathcal{J}(\boldsymbol\theta) &=\int _\mathcal{S}\rho^{\boldsymbol\pi}(s)\int _{\mathcal{A}^{i}}\bigtriangledown _{\theta _i}\pi_{\theta_i}(a^i|s) \\ 
&\int _{\mathcal{A}^{-i}}\prod_{j\neq i}^{n}\pi_{\theta_{j}}(a^{j}|s)Q_{\boldsymbol\pi}(s,\textbf{a})d\textbf{a}^{-i}da^ids,
\end{aligned}
\end{equation}
where the notation $-i$ indicates all other agents except agent $i$, $\mathcal{P}(s\to s^{\prime},t,\boldsymbol\pi)$ denotes the density at state $s^{\prime}$ after transitioning for $t$ time steps from state $s$, and $\rho^{\boldsymbol\pi}(s)=\int _\mathcal{S}\Sigma_{t=1}^\infty \gamma ^{t-1}\rho_0(s)\mathcal{P}(s\to s^{\prime},t,\boldsymbol\pi)$ is the (unnormalized) discounted distribution over states induced by the joint policy $\boldsymbol\pi$.

\section{Method}
This section considers an auto-regressive joint policy with fixed execution order $\{1,2,\dots,n\}$: 
\begin{equation}
\boldsymbol\pi(\textbf{a}|s)=\prod_{i=1}^{n}\pi_{\theta _{i}}(a^i|s, a^1,\dots,a^{i-1})
\end{equation}
Although such factorization introduces forward dependency among agents, it ignores the reaction of subsequent policy updates on the preceding actions. To achieve bidirectional dependency, we propose Back-Propagation Through Agents (BPTA) to pass gradients across agents. Specifically, we leverage the reparameterization trick and devise a new \textit{multi-agent conditional policy gradient theorem} that exploits the action dependency among agents. To cover any action-dependent policy, the relationship between the joint policy and individual policies can be stated as:
\begin{equation}
\label{eq:conditional}
\boldsymbol\pi(\textbf{a}|s)=\prod_{i=1}^{n}\pi_{\theta_{i}}(a^i|s,a^{\mathcal{F}_i})\;,
\end{equation}
where $\mathcal{F}_i$ denotes the set of agents on which agent $i$ has a forward dependency, and $a^{\mathcal{F}_i}$ are the actions taken by those agents.

\begin{figure*}[htbp!]
\centering
\includegraphics[width=0.7\textwidth]{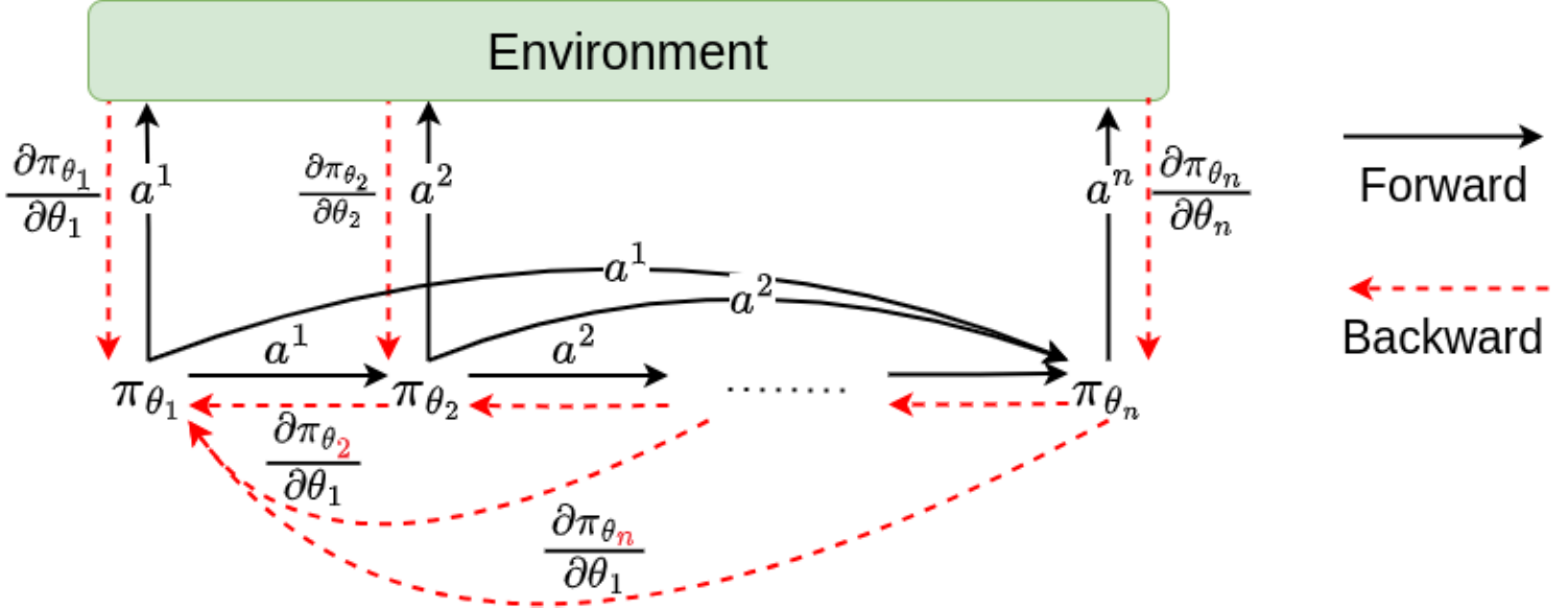} 
\caption{Learning framework of BPTA-augmented auto-regressive policy gradient algorithm. BPTA internally completes the causality loop by two processes: (1) the forward process, which generates the causal flow by action; and (2) the backward process, which inverts the causal flow by propagating the feedback.}
\label{bpta-ar}
\end{figure*}

\subsection{Back-Propagation Through Agents}
In the social context, joint action usually requires people to actively modify their own actions to reach a common action goal. Accordingly, joint action demands not only the integration of one’s own and others’ actions but also the corresponding consequences. However, most previous approaches assume that each agent only needs to account for its own learning process and completely disregarded the evaluation of its dependent actions result. In this section, we will show that our conditional gradient explicitly accounts for the effect of an agent's actions on the policies of its backward-dependent peer agents by additionally including agent feedback passed through the action dependency.

\begin{restatable}[Conditional Multi-Agent Stochastic Policy Gradient Theorem]{theorem}{conditional}
\label{conditional}
For any episodic cooperative stochastic game with n agents, the gradient of the expected total reward for agent $i$, who has a backward dependency on some other peer agents $\mathcal{B}_i$ using parameters $\theta_{\mathcal{B}_i}$, with respect to current policy parameters $\theta_i$ is:
\begin{equation}
\begin{aligned}
\label{eq:macpgt}
&\bigtriangledown _{\theta _i} \mathcal{J}(\boldsymbol\theta) =\int _\mathcal{S}\rho^{\boldsymbol\pi}(s)\Bigl[ \int _{\mathcal{A}^{i}}\underbrace{\bigtriangledown _{\theta _i}\pi_{\theta_i}(a^i|s,a^{\mathcal{F}_i})}_{\textnormal{Own Learning}} \\ 
&\qquad \int_{\mathcal{A}^{-i}}\pi_{\theta_{-i}}(a^{-i}|s^{\prime},a^{\mathcal{F}_{-i}}) \; + \; \int_{\mathcal{A}^i}\pi_{\theta_i}(a^i|s,a^{\mathcal{F}_i}) \\ 
&\qquad \int_{\mathcal{A}^{\mathcal{F}_i}}\pi_{\theta_{\mathcal{F}_i}}(a^{\mathcal{F}_i}|s,a^{\mathcal{F}_{\mathcal{F}_i} }) \\
&\int_{\mathcal{A}^{\mathcal{B}_i}}\underbrace{\bigtriangledown _{a^i} \pi_{\theta_{\mathcal{B}_i}}(a^{\mathcal{B}_i}|s,a^i,a^{\mathcal{F}_{\mathcal{B}_i} \setminus \{i\}})\bigtriangledown _{\theta_i}g(\theta_i, \varepsilon)}_{\textnormal{Peer Learning}}\Bigr] \\ 
&Q_{\boldsymbol\pi}(s,\textbf{a})d\textbf{a}^{-i}da^ids, 
\end{aligned}
\end{equation}
where $\mathcal{F}_{\mathcal{B}_i}$ indicates the set of agents on which $\mathcal{B}_i$ have forward dependencies.
\end{restatable}
\begin{proof}
See the Appendix for detailed proof.
\end{proof}

From Theorem \ref{conditional}, we note that the policy gradient for agent $i$ at each state has two primary terms. The first term $\bigtriangledown _{\theta _i}\pi_{\theta_i}(a^i|s,a^{\mathcal{F}_i})$ corresponds to the independent multi-agent policy gradient which explicitly differentiates through $\pi_{\theta_i}$ with respect to the current parameters $\theta_i$. This enables agent $i$ to model its own learning. By contrast, the second term $\bigtriangledown _{a^i} \pi_{\theta_{\mathcal{B}_i}}(a^{\mathcal{B}_i}|s,a^i,a^{\mathcal{F}_{\mathcal{B}_i} \setminus \{i\}})\bigtriangledown _{\theta_i}g(\theta_i, \varepsilon)$ aims to additionally account for how the consequences of the corresponding action on its backward dependent agents’ policies influence its direction of performance improvement. As a result, the peer learning term enables agents to adjust their own policies to those of action partners, which facilitates fast and accurate inter-agent coordination. Interestingly, the peer learning term, which evaluates the impact of an agent's actions on its peer agents, specifies auxiliary rewards for adapting its policy in accordance with these reactions. 

With Theorem \ref{conditional} and an auto-regressive joint policy, we are ready to present the learning framework of our BPTA-augmented auto-regressive policy gradient algorithm. As illustrated in Figure \ref{bpta-ar}, in the forward process, direct connections and skip connections \cite{he2015deep} connect the action of one predecessor agent to the input of subsequent agents, even those are not adjacent to it in execution order. As for the backward process described by dashed lines, in addition to the interactive feedback from the environment, there are alternative pathways provided by direct and skip connections, which allows successors to provide feedback to predecessors through gradients. Furthermore, these two types of processes are interleaved to allow for a causal flow loop within and across agents.

Our proposed algorithm can be conveniently integrated into most PG-based methods. Given the empirical performance and monotonic policy improvement of PPO \cite{schulman2017proximal}, we propose \textit{Bidirectional Proximal Policy Optimisation} (BPPO) to incorporate the proposed theorem with PPO. Following the sequential decision scheme, it is intuitive for BPPO to adopt the sequential update scheme \cite{wang2023order, kuba2022trust}, where the updates are performed in reverse order of the execution sequence. We provide comparisons of the simultaneous update scheme and the sequential update scheme in Appendix.

\begin{corollary}[Clipping Objective of BPPO]
Let $\boldsymbol\pi$ be an auto-regressive joint policy with fixed execution order $\{1,2,\dots,n\}$, and $  \tilde{\boldsymbol\pi}^{i+1:n}$ be the updated joint policy of agents set $\{i+1,\dots,n\}$. For brevity, we omit the preceding actions in the policy. Then the clipping objective of BPPO is:
\begin{equation}
\begin{aligned}
\label{eq:bppo}
&\mathbb{E}_{s\sim\rho^{\boldsymbol\pi}(s),\textbf{a} \sim \boldsymbol\pi, \varepsilon \sim p(\varepsilon)}\Bigl[\textnormal{min}\Bigl(\frac{\pi_{\theta_i}(a^i|s)}{\pi_{\theta_i}^{old}(a^i|s)}M^{i+1:n}+ \\
&\qquad \textnormal{detach}\Bigl(\frac{\pi_{\theta_i}(a^i|s)}{\pi_{\theta_i}^{old}(a^i|s)}\Bigr)\bigtriangledown _{a^i}M^{i+1:n}g(\theta_i, \varepsilon)\hat{A}(s,\textbf{a}), \\
&\qquad \textnormal{clip}\Bigl(\frac{\pi_{\theta_i}(a^i|s)}{\pi_{\theta_i}^{old}(a^i|s)},1\pm \epsilon \Bigr)M^{i+1:n}+ \\
&\qquad \textnormal{detach}\Bigl(\textnormal{clip}\Bigl(\frac{\pi_{\theta_i}(a^i|s)}{\pi_{\theta_i}^{old}(a^i|s)},1\pm \epsilon \Bigr)\Bigr) \\
&\qquad \bigtriangledown _{a^i}M^{i+1:n}g(\theta_i, \varepsilon)\hat{A}(s,\textbf{a})\Bigr)\Bigr],
\end{aligned}
\end{equation}
where $M^{i+1:n}=\frac{\tilde{\boldsymbol\pi}^{i+1:n}(a^{i+1:n}|s)}{\boldsymbol\pi^{i+1:n}(a^{i+1:n}|s)}$, $\bigtriangledown _{a^i}M^{i+1:n}=\bigtriangledown _{a^i}\frac{\tilde{\boldsymbol\pi}^{i+1:n}(a^{i+1:n}|s)}{\boldsymbol\pi^{i+1:n}(a^{i+1:n}|s)}$, and $\textnormal{detach}()$ represents detaching the input from the computation graph, meaning that the input will not contain gradients. $\hat{A}(s,\textbf{a})$ is an estimator of the advantage function $A(s,\textbf{a})$ computed by GAE \cite{schulman2018highdimensional}.
\end{corollary}

We provide the pseudo-code for BPPO in Algorithm \ref{alg:bppo}.
\begin{algorithm}[h!]
\caption{Bidirectional Proximal Policy Optimisation}
\label{alg:bppo}
\textbf{Initialize}: The auto-regressive joint policy $\boldsymbol\pi=\{\pi_{\theta_{1}},\dots,\pi_{\theta_{n}}\}$, the global value function $V$, replay buffer B, and the execution order $\{1,\dots,n\}$.\\
\begin{algorithmic}[1] 
\FOR{episode $k = 0,1,\dots$}
      \STATE Collect a set of trajectories by sequentially executing policies according to the execution order;
      \STATE Push data into B.
      \STATE Compute the advantage approximation $\hat{A}(s,\textbf{a})$ with GAE.
      \STATE Compute the value target $v(s)=\hat{A}(s,\textbf{a})+V(s)$.
      \STATE Set agent $i$'s gradient w.r.t. agent $j$'s action $\{c_i^j=0 \,|\, i \in \mathcal{N}, j \in \mathcal{N} \}$ and $M^{n+1}=1$.
      \FOR{agent $i = n,n-1,\dots,1$}
        \STATE Generate $g(\theta_i, \varepsilon)$ based on the reparameterization trick.
        \STATE Compute $\bigtriangledown _{a^i}M^{i+1:n}$ based on $\{c_j^i=0 \,|\, j \in \{i+1,\dots,n\} \}$ and the chain rule.
        \STATE Optimize Eq. \ref{eq:bppo} w.r.t. $\theta_{i}$.
        \FOR{agent $j = 1,\dots,i-1$}
            \STATE Compute the gradient $c$ of $\frac{\pi_{\theta_i}(a^i|s)}{\pi_{\theta_i}^{old}(a^i|s)}$ w.r.t. $a^{j}$.
            \STATE Set $c_i^j=c$.
        \ENDFOR
        \STATE Compute $M^{i:n}=\frac{\pi_{\theta_i}(a^i|s)}{\pi_{\theta_i}^{old}(a^i|s)}M^{i+1:n}$.
      \ENDFOR
      \STATE Update the value function by the following formula:
      \STATE $V=\textnormal{argmin}_{V} \mathbb{E}_{s\sim\rho^{\boldsymbol\pi}(s)}\Bigl[ \left\| v(s)-V(s) \right\|^2 \Bigr]$.
\ENDFOR
\end{algorithmic}
\end{algorithm}

\begin{figure*}[htbp!]
\centering
\includegraphics[width=0.9\textwidth,height=0.58\textwidth]{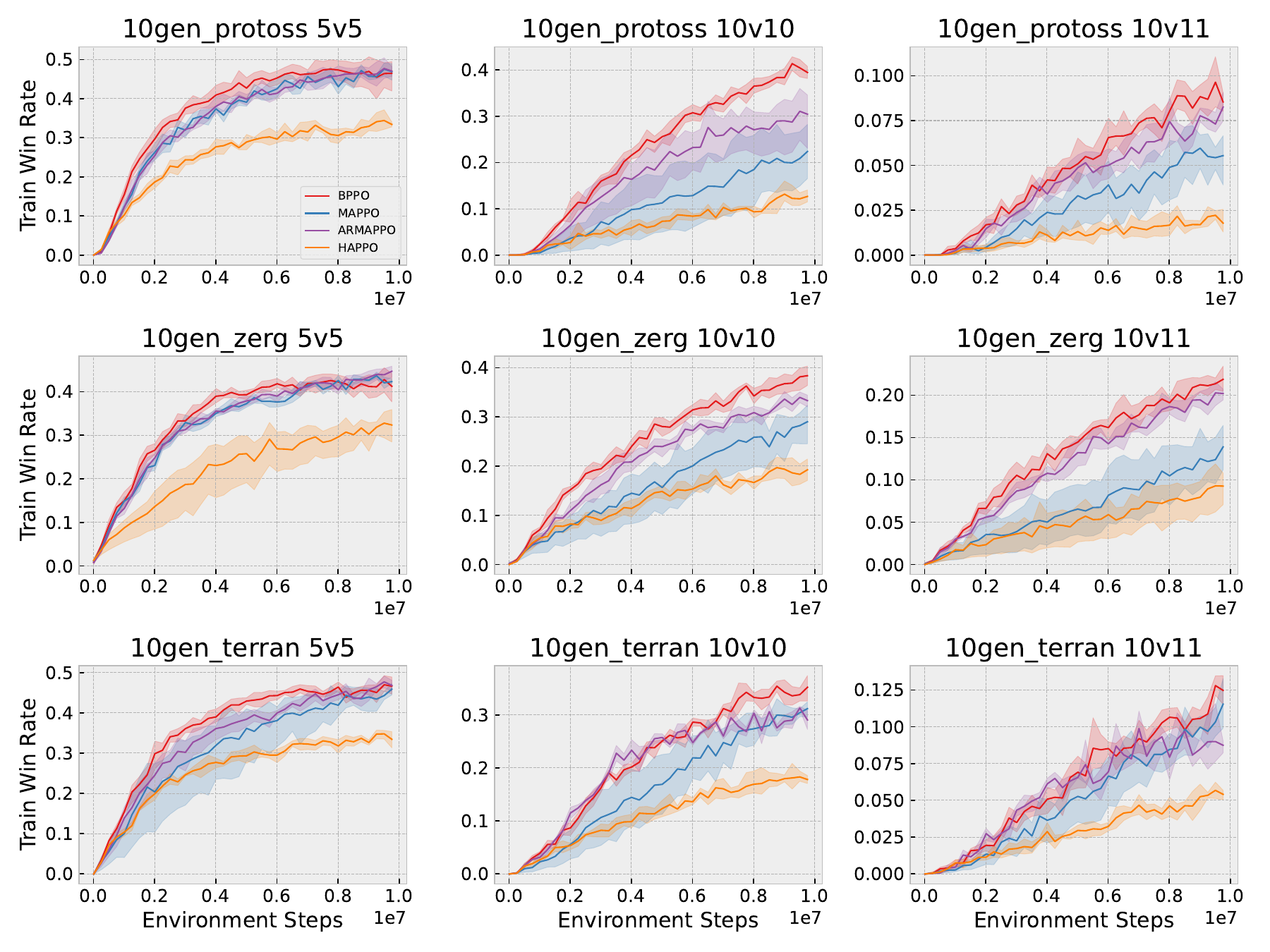} 
\caption{Comparison of training results on SMACv2.}
\label{smac-train}
\end{figure*}

\section{Experiments}
In this section, we experimentally evaluate BPPO on several multi-agent benchmarks, including two matrix games \cite{claus1998dynamics}, the StarCraft Multi-Agent Challenge Version 2 (SMACv2) \cite{ellis2022smacv2}, the Multi-agent MuJoCo (MA-MuJoCo) \cite{peng2021facmac}, and the Google Research Football (GRF) \cite{Kurach_Raichuk_Stańczyk_Zając_Bachem_Espeholt_Riquelme_Vincent_Michalski_Bousquet_Gelly_2020}, comparing them against MAPPO \cite{NEURIPS2022_9c1535a0}, HAPPO \cite{kuba2022trust}, and Auto-Regressive MAPPO (ARMAPPO) \cite{pmlr-v162-fu22d}. All results are presented using the mean and standard deviation of five random seeds. We fixed the execution order as sequential in all experiments. Additionally, we compared the effects of different execution orders in Appendix. More experimental details and results on these tasks are included in Appendix.

\begin{table}[ht]
    \centering
    \begin{tabular}{cc|ccc}
        \multicolumn{2}{c}{} & \multicolumn{3}{c}{Player 2} \\
        \multicolumn{2}{c}{} & A & B & C \\
        \cline{3-5}
        \parbox[t]{4mm}{\multirow{3}{*}{\rotatebox[origin=c]{90}{Player 1}}} & A &  \multicolumn{1}{c|}{11} & \multicolumn{1}{c|}{-30} & \multicolumn{1}{c|}{0} \\
        \cline{3-5}
        & B & \multicolumn{1}{c|}{-30} & \multicolumn{1}{c|}{7} & \multicolumn{1}{c|}{0} \\
        \cline{3-5}
        & C & \multicolumn{1}{c|}{0} & \multicolumn{1}{c|}{6} & \multicolumn{1}{c|}{5} \\
        \cline{3-5}
    \end{tabular}
    \caption{Payoff Matrix of the Climbing game.}
    \label{tab:climbing_payoff_matrix}
\end{table}
\begin{table}[ht]
    \centering
    \begin{tabular}{cc|ccc}
        \multicolumn{2}{c}{} & \multicolumn{3}{c}{Player 2} \\
        \multicolumn{2}{c}{} & A & B & C \\
        \cline{3-5}
        \parbox[t]{4mm}{\multirow{3}{*}{\rotatebox[origin=c]{90}{Player 1}}} & A &  \multicolumn{1}{c|}{-100} & \multicolumn{1}{c|}{0} & \multicolumn{1}{c|}{10} \\
        \cline{3-5}
        & B & \multicolumn{1}{c|}{0} & \multicolumn{1}{c|}{2} & \multicolumn{1}{c|}{0} \\
        \cline{3-5}
        & C & \multicolumn{1}{c|}{10} & \multicolumn{1}{c|}{0} & \multicolumn{1}{c|}{-100} \\
        \cline{3-5}
    \end{tabular}
    \caption{Payoff Matrix of the Penalty game.}
    \label{tab:penalty_payoff_matrix}
\end{table}

\begin{figure*}[htbp!]
\centering
\includegraphics[width=0.9\textwidth,height=0.6\linewidth]{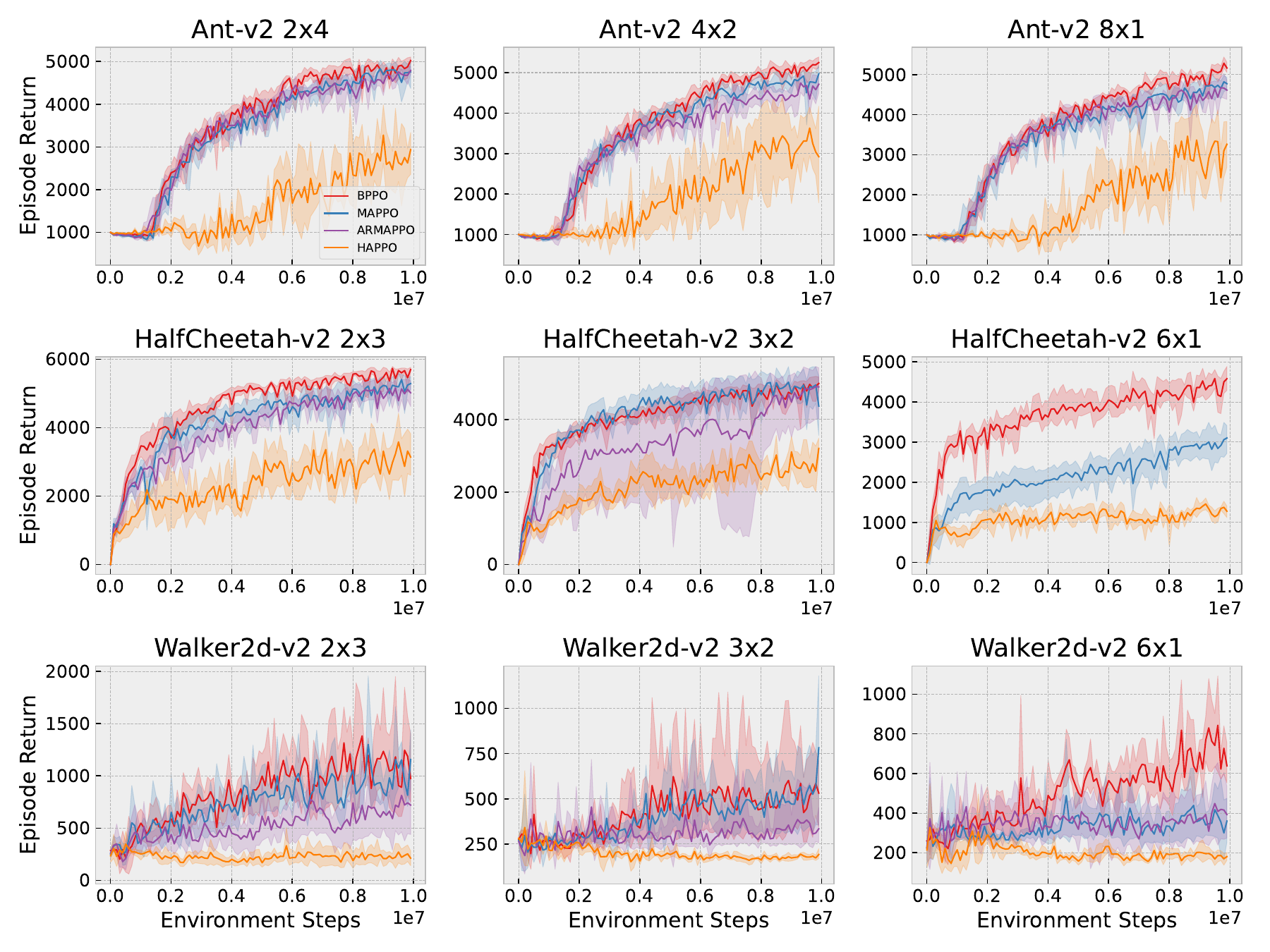} 
\caption{Performance comparison on multiple Multi-Agent MuJoCo tasks.}
\label{mujoco-main}
\end{figure*}

\subsection{Matrix Games}
\begin{figure}[htbp!]
\centering
\includegraphics[width=0.36\textwidth]{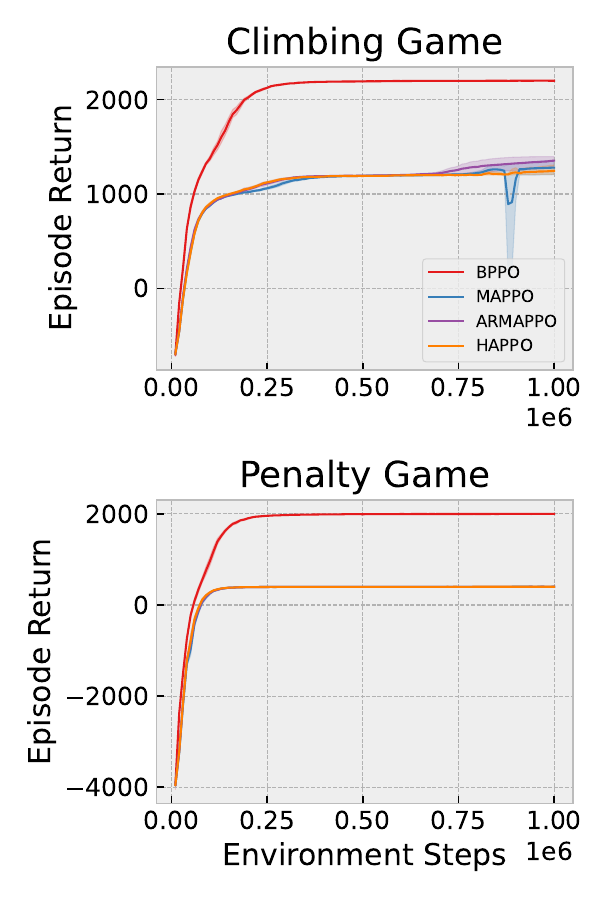} 
\caption{Performance comparison on the Climbing game and the Penalty game.}
\label{fig:matrix}
\end{figure}
As presented in Table \ref{tab:climbing_payoff_matrix} and \ref{tab:penalty_payoff_matrix}, the Climbing game and the Penalty game \cite{claus1998dynamics} are shared-reward multi-agent matrix games with two players in which each player has three actions at their disposal. The two matrix games have several Nash equilibria, but only one or two Pareto-optimal Nash equilibria \cite{christianos2022pareto}. Although stateless and with simple action space, the matrix games are difficult to solve as the agents need to coordinate among two optimal joint actions. Figure \ref{fig:matrix} shows that the compared baselines will converge to a locally optimal policy while BPPO is the only method that converges to the Pareto-optimal equilibria in all games. This is because BPPO explicitly considers the dependency success to find the optimal joint policy. The gap between the proposed method and the baselines is possibly due to that agents are fully independent of each other when making decisions in those methods. Interestingly, we observe that even with an auto-regressive policy, ARMAPPO still fails to find the optima. However, when we project the preceding actions inputted to each agent in ARMAPPO to higher-dimensional vectors, ARMAPPO w/ PROJ successfully converges to the optimal policy (verified in Appendix).

\subsection{SMACv2}

In SMAC, a group of learning agents aims to defeat the units of the enemy army controlled by the built-in heuristic AI. Despite its popularity, the SMAC is restricted to limited stochasticity \cite{ellis2022smacv2}. Compared to the SMAC, we instead evaluate our method on the more challenging SMAC-v2 benchmark which is designed with higher randomness. We evaluate our method on 3 maps (Zerg, Terran, and Protoss) with symmetric (20-vs-20) and asymmetric (10-vs-11 and 20-vs-23) units. As shown in Figure \ref{smac-train}, we generally observe that BPPO outperforms the baselines across most scenarios. In three 10 vs 10 scenarios, the margin between BPPO and the baselines becomes larger. Additionally, we also observe that BPPO has better stability as the variance shows.
\begin{figure*}[htbp!]
\centering
\includegraphics[width=0.99\textwidth]{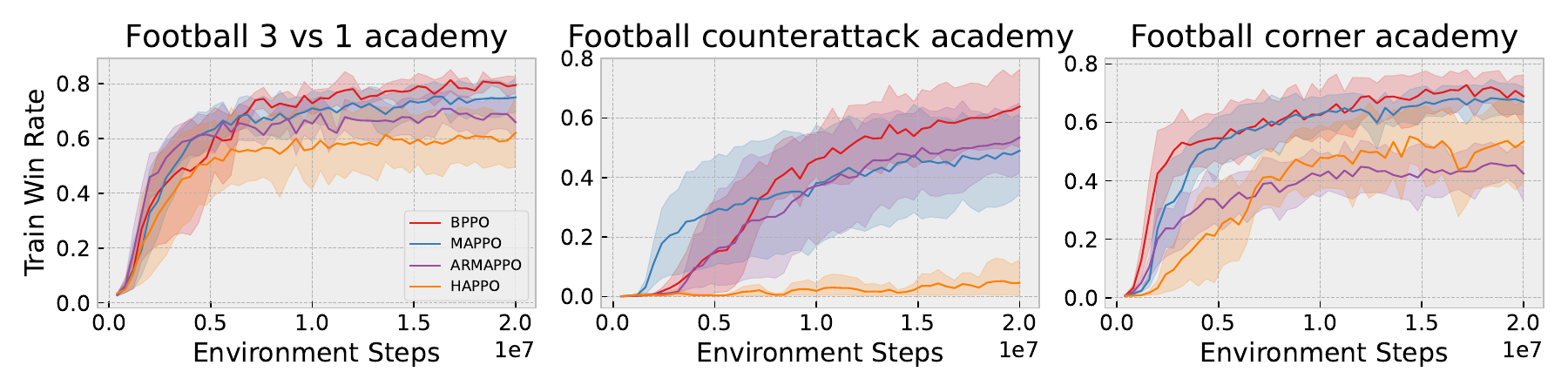} 
\caption{Averaged train win rate on the Google Research Football scenarios.}
\label{football-train}
\end{figure*}

\subsection{MA-MuJoCo}

Multi-Agent MuJoCo is a novel benchmark for decentralized cooperative continuous multi-agent robotic control in which single robots are decomposed into individual segments controlled by different agents. We show the performance comparison against the baselines in Figure \ref{mujoco-main}. We can see that BPPO achieves comparable performance compared to the baselines in most tasks while superiorly outperforming others in certain scenarios. It is also worth noting that the observed performance gap between BPPO and ARMAPPO can be attributed to the effectiveness of backward dependency. Meanwhile, we can observe that the performance gap between BPPO and its rivals enlarges with the increasing number of agents. Especially in HalfCheetah-v2 6x1 and Walker2d-v2 6x1, when other algorithms fail to learn any meaningful joint policies or converge to suboptimal points, BPPO outperforms the baselines by a large margin. Interestingly, especially in HalfCheetah 6x1 task, the performance of ARMAPPO even drops to negative. These results show that BPPO enables agents to achieve consistent joint improvement.  
\begin{figure}[htbp!]
\centering
\includegraphics[width=0.4\textwidth]{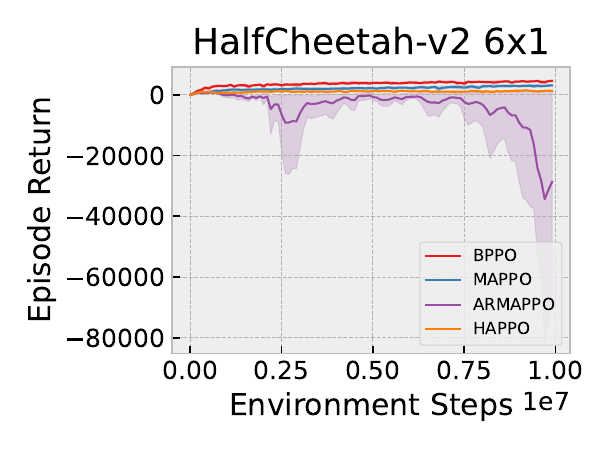} 
\caption{Performance comparison on the HalfCheetah-v2 6x1 Multi-Agent MuJoCo task. ARMAPPO performs poorly, with even negative rewards.}
\label{mujoco-half-ar}
\end{figure}

\subsection{GRF}

Google Research Football is a complex environment with large action space and sparse rewards where agents aim to score goals against fixed rule-based opponents. We evaluate our method in both GRF academy scenarios (3-vs-1 with keeper, corner, and counterattack hard) and full-game scenarios (5-vs-5). As can be seen in Figure \ref{football-train}, in the academy scenarios, only a minor difference can be observed between the proposed method and the baselines except for the counterattack task. Additionally, Figure \ref{football-5v5} shows that BPPO gains the highest score in the complex 5 vs 5 full-game scenario, while the baselines barely learn anything. Despite the negative performance of all algorithms, BPPO has achieved improved returns and is still learning compared to other algorithms that have consistently maintained their initial values. 
\begin{figure}[htbp!]
\centering
\includegraphics[width=0.4\textwidth]{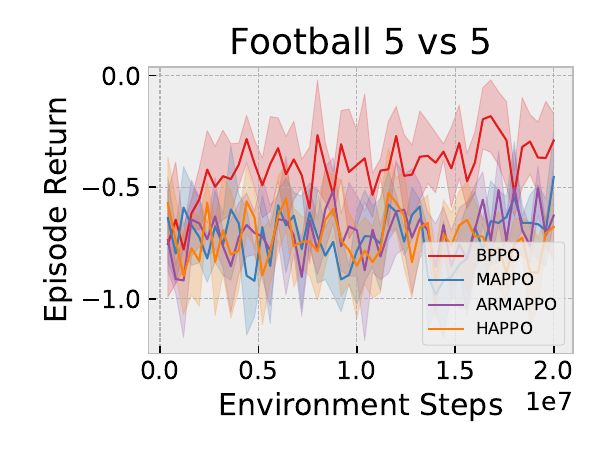} 
\caption{Comparisons of averaged return on the 5-vs-5 scenario.}
\label{football-5v5}
\end{figure}

\section{Conclusion}
In this paper, we propose Back-Propagation Through Agents (BPTA) to enable bidirectional dependency in any action-dependent multi-agent policy gradient (MAPG) methods. By conditional multi-agent stochastic policy gradient theorem, we can directly model both an agent’s own action effect and the feedback from its backward dependent agents. We evaluate the proposed Bidirectional Proximal Policy Optimisation (BPPO) based on BPTA and auto-regressive policy on several multiagent benchmarks. Results show that BPPO improves the performance against current state-of-the-art MARL methods. For future work, we plan to study the methods to learn the adaptive order.

\section*{Acknowledgments}
Zhiyuan Li acknowledges the financial support from the China Scholarship Council (CSC).

\appendix
\section{Derivation of Multi-Agent Conditional Policy Gradient Theorem}
\label{appendix:conditional}
\conditional*
\begin{proof}
The proof mainly follows the standard Stochastic Policy Gradient Theorem \cite{NIPS1999_464d828b}. We assume agent $i$ follows the policy $\pi_{\theta_i}(a^{i}|s, a^{\mathcal{F}_i})$ parameterized by $\theta_i$, where $\mathcal{F}_i$ are the agents on which agent $i$ has a forward dependency, and $a^{\mathcal{F}_i}$ are the actions taken by those agents drawn from the joint policy $\pi_{\theta_{\mathcal{F}_i}}$. We begin our derivation from the value function $V_{\boldsymbol\pi}(s)$ expand it with the state-action $Q_{\boldsymbol\pi}(s,\textbf{a})$ and the conditional joint policy (see Equation \ref{eq:conditional}):
\begin{equation}
\label{eq:expanded_V}
\begin{aligned}
V_{\boldsymbol\pi}(s) = &\int_{\mathcal{A}^1}\pi_{\theta_1}(a^1|s,a^{\mathcal{F}_1})\cdots \\ 
&\int_{\mathcal{A}^n}\pi_{\theta_n}(a^n|s,a^{\mathcal{F}_n})Q_{\boldsymbol\pi}(s,\textbf{a})d\textbf{a}
\end{aligned}
\end{equation}

In Equation \ref{eq:expanded_V}, we note that the individual policies of all agents except agent $i$ may be conditioned on the action of agent $i$, i.e. $i \in \mathcal{F}_j, \forall j \in \mathcal{N}\setminus \{i\}$. Importantly, the optimization of those agents that have forward dependencies on agent $i$ will result in backward dependencies between them and agent $i$ through the channel (the action of agent $i$). Thus, we denote $\mathcal{B}_i$ as the set of $m$ peer agents on which agent $i$ has a backward dependency and $\pi_{\theta_{\mathcal{B}_i}}$ as the joint policy of those agents.

Using the new notations, we continue from Equation \ref{eq:expanded_V} and derive the gradient of the value function with respect to $\theta_i$ by applying the product rule:
\begin{equation}
\begin{aligned}
\label{eq:derivitive_V}
& \bigtriangledown _{\theta _i}V_{\boldsymbol\pi}(s) 
\\
= & \bigtriangledown _{\theta _i}\int_{\mathcal{A}^1}\pi_{\theta_1}(a^1|s,a^{\mathcal{F}_1})\cdots 
\\
&\qquad \qquad \qquad  \int_{\mathcal{A}^n}\pi_{\theta_n}(a^n|s,a^{\mathcal{F}_n})Q_{\boldsymbol\pi}(s,\textbf{a})d\textbf{a} 
\\
= & \bigtriangledown _{\theta _i}\int_{\mathcal{A}^{\mathcal{F}_i}}\pi_{\theta_{\mathcal{F}_i}}(a^{\mathcal{F}_i}|s,a^{\mathcal{F}_{\mathcal{F}_i} })
\int_{\mathcal{A}^i}\pi_{\theta_i}(a^i|s,a^{\mathcal{F}_i}) 
\\
&\qquad \int_{\mathcal{A}^{\mathcal{B}_i}}\pi_{\theta_{\mathcal{B}_i}}(a^{\mathcal{B}_i}|s,a^i,a^{\mathcal{F}_{\mathcal{B}_i} \setminus \{i\}})Q_{\boldsymbol\pi}(s,\textbf{a})d\textbf{a}
\\
=& \int_{\mathcal{A}^{\mathcal{F}_i}}\pi_{\theta_{\mathcal{F}_i}}(a^{\mathcal{F}_i}|s,a^{\mathcal{F}_{\mathcal{F}_i} })\left [  \int_{\mathcal{A}^i}\bigtriangledown _{\theta _i} \pi_{\theta_i}(a^i|s,a^{\mathcal{F}_i})\right ] \\ 
&\underbrace{\qquad\int_{\mathcal{A}^{\mathcal{B}_i}}\pi_{\theta_{\mathcal{B}_i}}(a^{\mathcal{B}_i}|s,a^i,a^{\mathcal{F}_{\mathcal{B}_i} \setminus \{i\}})Q_{\boldsymbol\pi}(s,\textbf{a})d\textbf{a}}_{\text{Term A}} \; + \\
& \int_{\mathcal{A}^{\mathcal{F}_i}}\pi_{\theta_{\mathcal{F}_i}}(a^{\mathcal{F}_i}|s,a^{\mathcal{F}_{\mathcal{F}_i} })\int_{\mathcal{A}^i}\pi_{\theta_i}(a^i|s,a^{\mathcal{F}_i})\\ 
& \underbrace{\left [  \int_{\mathcal{A}^{\mathcal{B}_i}}\bigtriangledown _{\theta _i} \pi_{\theta_{\mathcal{B}_i}}(a^{\mathcal{B}_i}|s,a^i,a^{\mathcal{F}_{\mathcal{B}_i} \setminus \{i\}})\right ]Q_{\boldsymbol\pi}(s,\textbf{a})d\textbf{a}}_{\text{Term B}} \; + \\
& \int_{\mathcal{A}^{\mathcal{F}_i}}\pi_{\theta_{\mathcal{F}_i}}(a^{\mathcal{F}_i}|s,a^{\mathcal{F}_{\mathcal{F}_i} })\int_{\mathcal{A}^i}\pi_{\theta_i}(a^i|s,a^{\mathcal{F}_i}) \\
& \underbrace{\int_{\mathcal{A}^{\mathcal{B}_i}}\pi_{\theta_{\mathcal{B}_i}}(a^{\mathcal{B}_i}|s,a^i,a^{\mathcal{F}_{\mathcal{B}_i} \setminus \{i\}})\left [ \bigtriangledown _{\theta _i} Q_{\boldsymbol\pi}(s,\textbf{a})\right ]d\textbf{a}}_{\text{Term C}}.
\end{aligned}
\end{equation}

We first focus on the derivative of the joint policy $\pi_{\theta_{\mathcal{B}_i}}$ on which agent $i$ has a backward dependency in Term B:
\begin{equation}
\begin{aligned}
\label{eq:derivitive_bwd}
& \int_{\mathcal{A}^{\mathcal{B}_i}}\bigtriangledown _{\theta _i} \pi_{\theta_{\mathcal{B}_i}}(a^{\mathcal{B}_i}|s,a^i,a^{\mathcal{F}_{\mathcal{B}_i} \setminus \{i\}}) \\
=& \bigtriangledown _{\theta _i} \Bigl [  \int_{\mathcal{A}^{\mathcal{B}_i[1]}}\pi_{\theta_{\mathcal{B}_i[1]}}(a^{\mathcal{B}_i[1]}|s,a^i,a^{\mathcal{F}_{\mathcal{B}_i[1]} \setminus \{i\}})\cdots \\
&\qquad \int_{\mathcal{A}^{\mathcal{B}_i[m]}}\pi_{\theta_{\mathcal{B}_i[m]}}(a^{\mathcal{B}_i[m]}|s,a^i,a^{\mathcal{F}_{\mathcal{B}_i[m]} \setminus \{i\}})\Bigr ] \\
=& \left [  \int_{\mathcal{A}^{\mathcal{B}_i[1]}}\bigtriangledown _{\theta _i}\pi_{\theta_{\mathcal{B}_i[1]}}(a^{\mathcal{B}_i[1]}|s,a^i,a^{\mathcal{F}_{\mathcal{B}_i[1]} \setminus \{i\}})\right ] \\ 
&\quad \Pi_{\forall k \in \mathcal{B}_i \setminus \{1\} }\int_{\mathcal{A}^{\mathcal{B}_i[k]}}\pi_{\theta_{\mathcal{B}_i[k]}}(a^{\mathcal{B}_i[k]}|s,a^i,a^{\mathcal{F}_{\mathcal{B}_i[k]} \setminus \{i\}}) \; + \\
& \left [  \int_{\mathcal{A}^{\mathcal{B}_i[2]}}\bigtriangledown _{\theta _i}\pi_{\theta_{\mathcal{B}_i[2]}}(a^{\mathcal{B}_i[2]}|s,a^i,a^{\mathcal{F}_{\mathcal{B}_i[2]} \setminus \{i\}})\right ] \\ 
&\quad \Pi_{\forall k \in \mathcal{B}_i \setminus \{2\} }\int_{\mathcal{A}^{\mathcal{B}_i[k]}}\pi_{\theta_{\mathcal{B}_i[k]}}(a^{\mathcal{B}_i[k]}|s,a^i,a^{\mathcal{F}_{\mathcal{B}_i[k]} \setminus \{i\}}) \; + \\
&\qquad\qquad\qquad\qquad\qquad\cdots  \; + \\
& \left [  \int_{\mathcal{A}^{\mathcal{B}_i[m]}}\bigtriangledown _{\theta _i}\pi_{\theta_{\mathcal{B}_i[m]}}(a^{\mathcal{B}_i[m]}|s,a^i,a^{\mathcal{F}_{\mathcal{B}_i[m]} \setminus \{i\}})\right ] \\ 
&\quad \Pi_{\forall k \in \mathcal{B}_i \setminus \{m\} }\int_{\mathcal{A}^{\mathcal{B}_i[k]}}\pi_{\theta_{\mathcal{B}_i[k]}}(a^{\mathcal{B}_i[k]}|s,a^i,a^{\mathcal{F}_{\mathcal{B}_i[k]} \setminus \{i\}})
\end{aligned}
\end{equation}

To back-propagate gradient through actions, we use the reparameterization trick for continuous actions and the Gumbel-softmax trick for discrete actions, i.e., $a^i = g(\theta_i, \varepsilon)$, where $\varepsilon \sim p(\varepsilon)$. Then, Equation \ref{eq:derivitive_bwd} can be further expressed as:
\begin{equation}
\begin{aligned}
\label{eq:derivitive_bwd2}
&  \int_{\mathcal{A}^{\mathcal{B}_i[1]}}\bigtriangledown _{a^i}\pi_{\theta_{\mathcal{B}_i[1]}}(a^{\mathcal{B}_i[1]}|s,a^i,a^{\mathcal{F}_{\mathcal{B}_i[1]} \setminus \{i\}})\bigtriangledown _{\theta_i}g(\theta_i, \varepsilon) \\ 
&\quad \Pi_{\forall k \in \mathcal{B}_i \setminus \{1\} }\int_{\mathcal{A}^{\mathcal{B}_i[k]}}\pi_{\theta_{\mathcal{B}_i[k]}}(a^{\mathcal{B}_i[k]}|s,a^i,a^{\mathcal{F}_{\mathcal{B}_i[k]} \setminus \{i\}}) \; + \\
& \int_{\mathcal{A}^{\mathcal{B}_i[2]}}\bigtriangledown _{a^i}\pi_{\theta_{\mathcal{B}_i[2]}}(a^{\mathcal{B}_i[2]}|s,a^i,a^{\mathcal{F}_{\mathcal{B}_i[2]} \setminus \{i\}})\bigtriangledown _{\theta_i}g(\theta_i, \varepsilon) \\ 
&\quad \Pi_{\forall k \in \mathcal{B}_i \setminus \{2\} }\int_{\mathcal{A}^{\mathcal{B}_i[k]}}\pi_{\theta_{\mathcal{B}_i}}(a^{\mathcal{B}_i[k]}|s,a^i,a^{\mathcal{F}_{\mathcal{B}_i[k]} \setminus \{i\}}) \; + \\
&\qquad\qquad\qquad\qquad\qquad\cdots  \; + \\
&  \int_{\mathcal{A}^{\mathcal{B}_i}}\bigtriangledown _{a^i}\pi_{\theta_{\mathcal{B}_i[m]}}(a^{\mathcal{B}_i[m]}|s,a^i,a^{\mathcal{F}_{\mathcal{B}_i[m]} \setminus \{i\}})\bigtriangledown _{\theta_i}g(\theta_i, \varepsilon) \\ 
&\quad \Pi_{\forall k \in \mathcal{B}_i \setminus \{m\} }\int_{\mathcal{A}^{\mathcal{B}_i[k]}}\pi_{\theta_{\mathcal{B}_i[k]}}(a^{\mathcal{B}_i[k]}|s,a^i,a^{\mathcal{F}_{\mathcal{B}_i[k]} \setminus \{i\}}) \\
=& \int_{\mathcal{A}^{\mathcal{B}_i}}\bigtriangledown _{a^i} \pi_{\theta_{\mathcal{B}_i}}(a^{\mathcal{B}_i}|s,a^i,a^{\mathcal{F}_{\mathcal{B}_i} \setminus \{i\}})\bigtriangledown _{\theta_i}g(\theta_i, \varepsilon)
\end{aligned}
\end{equation}

Coming back to Term A, C in Equation \ref{eq:derivitive_V}, repeatedly unrolling the derivative of the Q-function by following \citeauthor{sutton2018reinforcement} yields:
\begin{equation}
\begin{aligned}
\label{eq:derivitive_i}
& \int_{\mathcal{S}}\sum_{t=0}^{\infty}\gamma ^{t}\mathcal{P}(s\to s^{\prime},t,\boldsymbol\pi)\int_{\mathcal{A}^i}\bigtriangledown _{\theta _i} \pi_{\theta_i}(a^i|s^{\prime},a^{\mathcal{F}_i}) \\
&\qquad \int_{\mathcal{A}^{-i}}\pi_{\theta_{-i}}(a^{-i}|s^{\prime},a^{\mathcal{F}_{-i}})Q_{\boldsymbol\pi}(s,\textbf{a})d\textbf{a}^{-i}da^ids^{\prime}
\end{aligned}
\end{equation}

Summarize Equations \ref{eq:derivitive_bwd2} and \ref{eq:derivitive_i} together:
\begin{equation}
\begin{aligned}
\label{eq:derivitive_V2}
& \bigtriangledown _{\theta _i}V_{\boldsymbol\pi}(s) \\
=& \int_{\mathcal{S}}\sum_{t=0}^{\infty}\gamma ^{t}\mathcal{P}(s\to s^{\prime},t,\boldsymbol\pi)\Bigl [ \int_{\mathcal{A}^i}\bigtriangledown _{\theta _i} \pi_{\theta_i}(a^i|s^{\prime},a^{\mathcal{F}_i}) \\
&\qquad \int_{\mathcal{A}^{-i}}\pi_{\theta_{-i}}(a^{-i}|s^{\prime},a^{\mathcal{F}_{-i}}) + \int_{\mathcal{A}^i}\pi_{\theta_i}(a^i|s,a^{\mathcal{F}_i}) \\
&\qquad \int_{\mathcal{A}^{\mathcal{F}_i}}\pi_{\theta_{\mathcal{F}_i}}(a^{\mathcal{F}_i}|s,a^{\mathcal{F}_{\mathcal{F}_i} }) \\ &\qquad \int_{\mathcal{A}^{\mathcal{B}_i}}\bigtriangledown _{a^i} \pi_{\theta_{\mathcal{B}_i}}(a^{\mathcal{B}_i}|s,a^i,a^{\mathcal{F}_{\mathcal{B}_i} \setminus \{i\}})\bigtriangledown _{\theta_i}g(\theta_i, \varepsilon)\Bigr ] \\ 
&\qquad Q_{\boldsymbol\pi}(s,\textbf{a})d\textbf{a}^{-i}da^ids^{\prime}
\end{aligned}
\end{equation}

Finally, we take the expectation over the possible initial states:
\begin{equation}
\begin{aligned}
\label{eq:derivative_macpgt}
\bigtriangledown _{\theta _i}\mathcal{J}(\boldsymbol\theta) &=\int _\mathcal{S}\rho^{\boldsymbol\pi}(s)\Bigl[ \int _{\mathcal{A}^{i}}\underbrace{\bigtriangledown _{\theta _i}\pi_{\theta_i}(a^i|s,a^{\mathcal{F}_i})}_{\textnormal{Own Learning}} \\ 
&\int_{\mathcal{A}^{-i}}\pi_{\theta_{-i}}(a^{-i}|s^{\prime},a^{\mathcal{F}_{-i}}) \; + \; \int_{\mathcal{A}^i}\pi_{\theta_i}(a^i|s,a^{\mathcal{F}_i}) \\ 
&\int_{\mathcal{A}^{\mathcal{F}_i}}\pi_{\theta_{\mathcal{F}_i}}(a^{\mathcal{F}_i}|s,a^{\mathcal{F}_{\mathcal{F}_i} }) \\
&\int_{\mathcal{A}^{\mathcal{B}_i}}\underbrace{\bigtriangledown _{a^i} \pi_{\theta_{\mathcal{B}_i}}(a^{\mathcal{B}_i}|s,a^i,a^{\mathcal{F}_{\mathcal{B}_i} \setminus \{i\}})\bigtriangledown _{\theta_i}g(\theta_i, \varepsilon)}_{\textnormal{Peer Learning}}\Bigr] \\ 
&Q_{\boldsymbol\pi}(s,\textbf{a})d\textbf{a}^{-i}da^ids 
\end{aligned}
\end{equation}
\end{proof}

\section{Experimental Details}
For a fair comparison, the implementation of BPPO and the baselines are based on the implementation of MAPPO. We keep all hyper-parameters unchanged at the origin best-performing status. The proposed method and compared baselines are implemented into parameter independent version. The common and different hyper-parameters used for MAPPO, HAPPO, ARMAPPO, and BPPO across all domains are listed in Table \ref{tab:common-hyper}-\ref{tab:common-hyper GRF} respectively.
\begin{table}[ht]
    \centering
    \begin{tabular}{c|c}
        \hline
        Parameter & Value \\
        \hline
        \multicolumn{2}{c}{Training} \\
        \hline
        optimizer & Adam \\
        optimizer epsilon & 1e-5 \\
        weight decay & 0 \\
        max grad norm & 10 \\
        data chunk length & 1 \\
        \hline
        \multicolumn{2}{c}{Model} \\
        \hline
        activation & ReLU \\
        \hline
        \multicolumn{2}{c}{PPO} \\
        \hline
        ppo-clip & 0.2 \\
        gamma & 0.99 \\
        gae lamda & 0.95 \\
        \hline
    \end{tabular}
    \caption{Common hyper-parameters used across all domains.}
    \label{tab:common-hyper}
\end{table}

\subsection{Matrix Games}
We list the hyper-parameters used in matrix games in Table.
\begin{table}[ht]
    \centering
    \begin{tabular}{c|c}
        \hline
        Parameter & Value \\
        \hline
        \multicolumn{2}{c}{Training} \\
        \hline
        actor lr & 5e-4 \\
        critic lr & 5e-4 \\
        entropy coef & 0.01 \\
        \hline
        \multicolumn{2}{c}{Model} \\
        \hline
        hidden layer & 1 \\
        hidden layer dim & 64 \\
        \hline
        \multicolumn{2}{c}{PPO} \\
        \hline
        ppo epoch & 15 \\
        ppo-clip & 0.2 \\
        num mini-batch & 1 \\
        \hline
        \multicolumn{2}{c}{Sample} \\
        \hline
        environment steps & 1000000 \\
        rollout threads & 50 \\
        episode length & 200 \\
        \hline
    \end{tabular}
    \caption{Common hyper-parameters used in matrix games.}
    \label{tab:common-hyper matrix}
\end{table}

\subsection{SMACv2}
We list the hyper-parameters used for each map of SMACv2 in Table.
\begin{table}[ht]
    \centering
    \begin{tabular}{c|c}
        \hline
        Parameter & Value \\
        \hline
        \multicolumn{2}{c}{Training} \\
        \hline
        actor lr & 5e-4 \\
        critic lr & 5e-4 \\
        entropy coef & 0.01 \\
        \hline
        \multicolumn{2}{c}{Model} \\
        \hline
        hidden layer & 1 \\
        hidden layer dim & 64 \\
        \hline
        \multicolumn{2}{c}{PPO} \\
        \hline
        ppo epoch & 5 \\
        ppo-clip & 0.2 \\
        num mini-batch & 1 \\
        \hline
        \multicolumn{2}{c}{Sample} \\
        \hline
        environment steps & 10000000 \\
        rollout threads & 50 \\
        episode length & 200 \\
        \hline
    \end{tabular}
    \caption{Common hyper-parameters used in the SMACv2.}
    \label{tab:common-hyper smacv2}
\end{table}

\subsection{MA-MuJoCo}
The hyper-parameters used for each task of MA-MuJoCo are listed in Table.
\begin{table}[ht]
    \centering
    \begin{tabular}{c|c}
        \hline
        Parameter & Value \\
        \hline
        \multicolumn{2}{c}{Training} \\
        \hline
        actor lr & 3e-4 \\
        critic lr & 3e-4 \\
        entropy coef & 0 \\
        \hline
        \multicolumn{2}{c}{Model} \\
        \hline
        hidden layer & 2 \\
        hidden layer dim & 64 \\
        \hline
        \multicolumn{2}{c}{PPO} \\
        \hline
        ppo epoch & 5 \\
        ppo-clip & 0.2 \\
        num mini-batch & 1 \\
        \hline
        \multicolumn{2}{c}{Sample} \\
        \hline
        environment steps & 10000000 \\
        rollout threads & 40 \\
        episode length & 100 \\
        \hline
        \multicolumn{2}{c}{Environment} \\
        \hline
        agent obsk & 0 \\
        \hline
    \end{tabular}
    \caption{Common hyper-parameters used in the MA-MuJoCo.}
    \label{tab:common-hyper MA-MuJoCo}
\end{table}
\begin{table}[ht]
    \centering
    \begin{tabular}{c|c}
        \hline
        Parameter & Value \\
        \hline
        \multicolumn{2}{c}{Environment} \\
        \hline
        agent obsk & None \\
        \hline
    \end{tabular}
    \caption{Different hyper-parameters used in the Ant task.}
    \label{tab:diff-hyper MA-MuJoCo}
\end{table}

\subsection{GRF}
We list the hyper-parameters used for each scenario of GRF in Table.
\begin{table}[ht]
    \centering
    \begin{tabular}{c|c}
        \hline
        Parameter & Value \\
        \hline
        \multicolumn{2}{c}{Training} \\
        \hline
        actor lr & 5e-4 \\
        critic lr & 5e-4 \\
        entropy coef & 0.01 \\
        \hline
        \multicolumn{2}{c}{Model} \\
        \hline
        hidden layer & 1 \\
        hidden layer dim & 64 \\
        \hline
        \multicolumn{2}{c}{PPO} \\
        \hline
        ppo epoch & 15 \\
        ppo-clip & 0.2 \\
        num mini-batch & 2 \\
        \hline
        \multicolumn{2}{c}{Sample} \\
        \hline
        environment steps & 20000000 \\
        rollout threads & 50 \\
        episode length & 200 \\
        \hline
    \end{tabular}
    \caption{Common hyper-parameters used in the GRF.}
    \label{tab:common-hyper GRF}
\end{table}

\section{Additional Results}
\subsection{Multi-agent MuJoCo}
To further verify the effectiveness of our method, we add a comparison with a coordinated graph method GCS \cite{10.5555/3535850.3535976}. The experimental result is shown in Figure \ref{mujoco-half-ar-app}.
\begin{figure}[htbp!]
\centering
\includegraphics[width=0.35\textwidth]{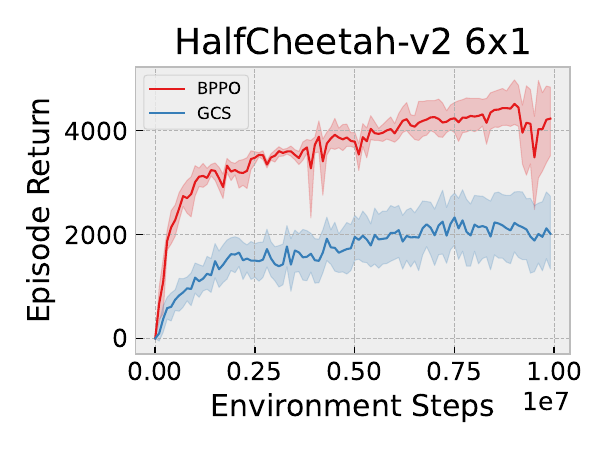} 
\caption{Performance comparison on the HalfCheetah-v2 6x1 Multi-Agent MuJoCo task.}
\label{mujoco-half-ar-app}
\end{figure}

\subsection{SMACv2}
Here, we evaluate the test win rate of BPPO and the baseline on selected SMACv2 maps. Results are shown in Figure \ref{smac-main}. In general, BPPO matches or slightly outperforms the best performance of the baselines on all nine maps.
\begin{figure*}[htbp!]
\centering
\includegraphics[width=0.9\textwidth]{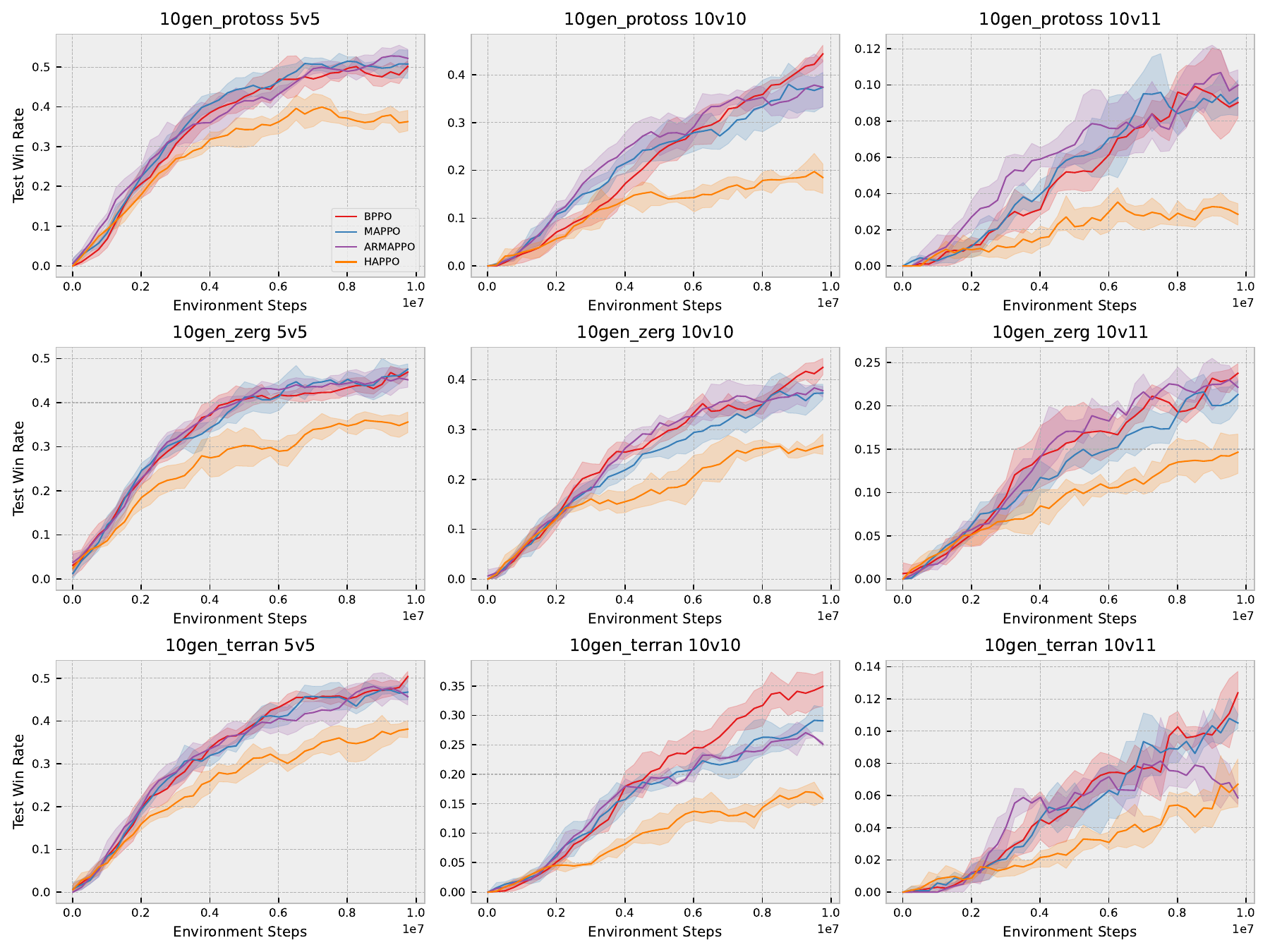} 
\caption{Comparison of the mean test win rate on SMACv2.}
\label{smac-main}
\end{figure*}
\subsection{GRF}
We further compare the evaluation performance of BPPO and the baselines in Figure \ref{football-main}. In all scenarios, BPPO can achieve comparable or better performances. 
\begin{figure*}[htbp!]
\centering
\includegraphics[width=0.9\textwidth]{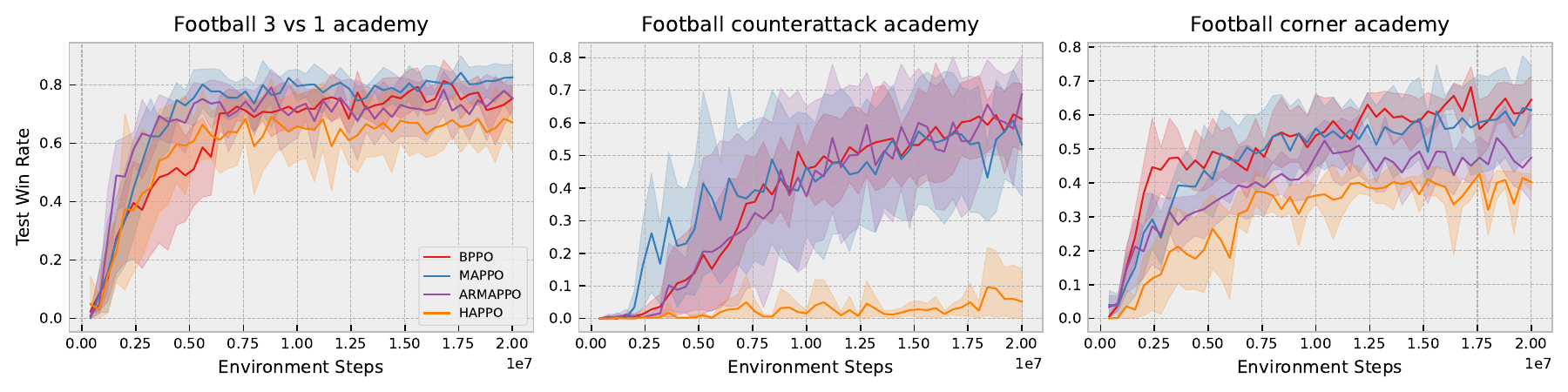} 
\caption{Averaged test win rate on the Google Research Football scenarios.}
\label{football-main}
\end{figure*}

\subsection{Ablation Study}
\subsubsection{Training Scheme}
Figure \ref{fig:mujoco-ablation} shows the effect of different training schemes on BPPO: 1) the simultaneous update scheme and 2) the sequential update scheme. We observe that the sequential and simultaneous update schemes achieve similar performance.
\begin{figure*}[htbp!]
\centering
\includegraphics[width=0.8\textwidth]{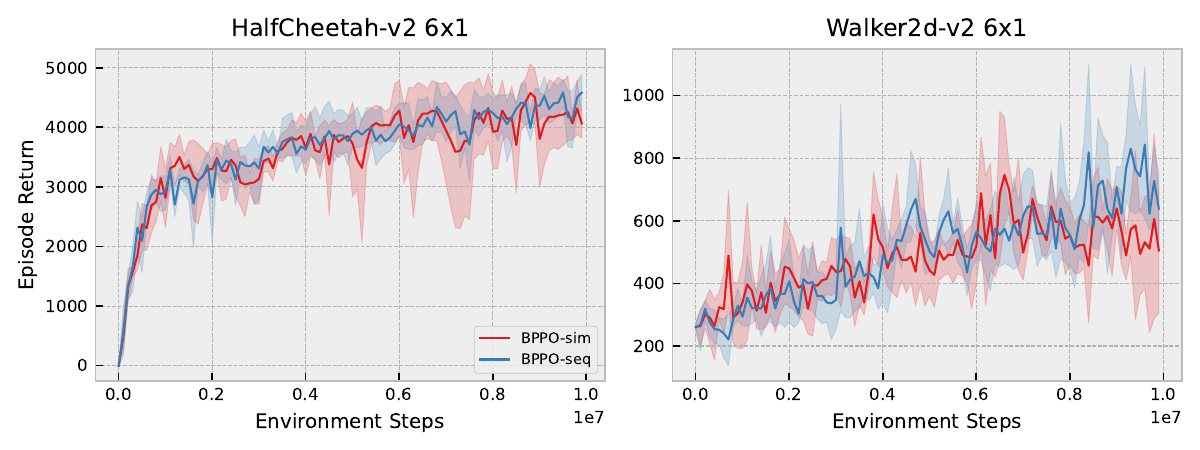} 
\caption{Ablation experiments on training schemes.}
\label{fig:mujoco-ablation}
\end{figure*}

\subsubsection{Execution Order}
To assess the impact of execution order, we opted for three different execution orders: sequential execution, reverse execution, and random execution. Figure  \ref{mujoco-half-ar2-app} shows the result:
\begin{figure}[htbp!]
\centering
\includegraphics[width=0.35\textwidth]{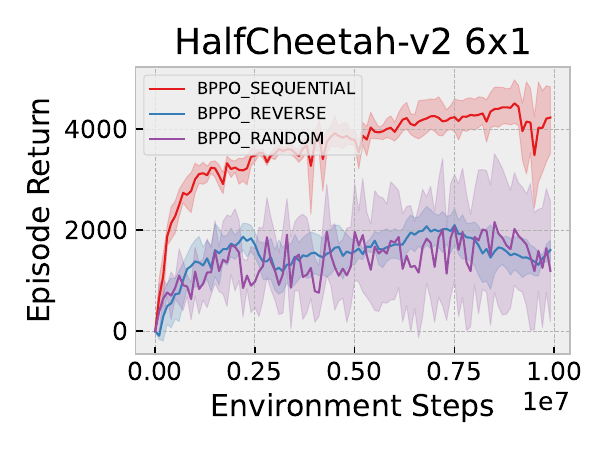} 
\caption{The effect of execution order on BPPO in HalfCheetah-v2 6x1 Multi-Agent MuJoCo task.}
\label{mujoco-half-ar2-app}
\end{figure}

\subsubsection{High-dimensional Projection}
Figure \ref{fig:matrix-ablation} shows the effect of high-dimensional projection on ARMAPPO.
\begin{figure}[htbp!]
\centering
\includegraphics[width=0.35\textwidth]{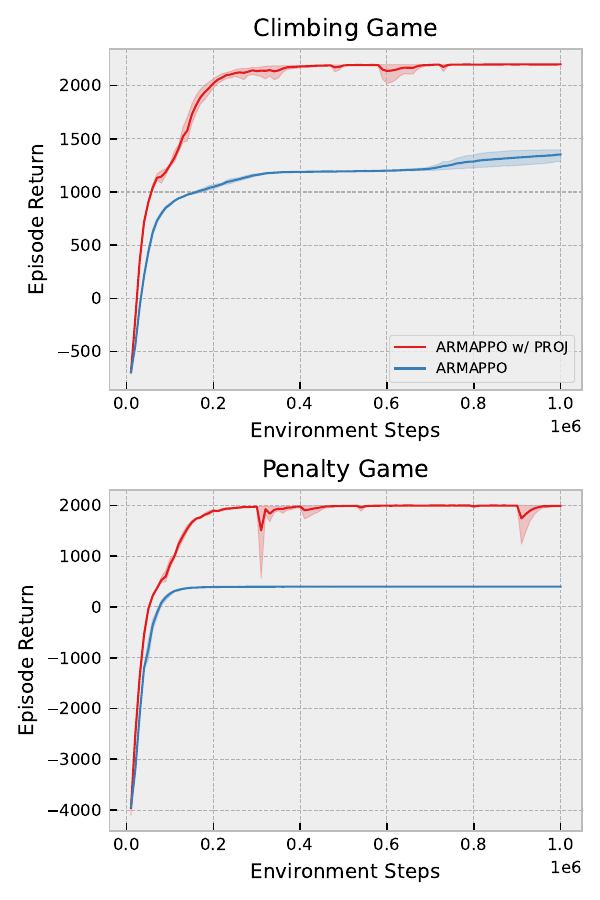} 
\caption{Ablation study on projection in matrix games.}
\label{fig:matrix-ablation}
\end{figure}

\bibliography{aaai24}

\begin{thebibliography}{45}
\providecommand{\natexlab}[1]{#1}

\bibitem[{Bertsekas(2021)}]{9317713}
Bertsekas, D. 2021.
\newblock Multiagent Reinforcement Learning: Rollout and Policy Iteration.
\newblock \emph{IEEE/CAA Journal of Automatica Sinica}, 8(2): 249--272.

\bibitem[{Box et~al.(2015)Box, Jenkins, Reinsel, and Ljung}]{box2015time}
Box, G.; Jenkins, G.; Reinsel, G.; and Ljung, G. 2015.
\newblock \emph{Time Series Analysis: Forecasting and Control}.
\newblock Wiley Series in Probability and Statistics. Wiley.
\newblock ISBN 9781118674925.

\bibitem[{Cho et~al.(2014)Cho, van Merrienboer, G{\"{u}}l{\c{c}}ehre, Bahdanau,
  Bougares, Schwenk, and Bengio}]{DBLP:conf/emnlp/ChoMGBBSB14}
Cho, K.; van Merrienboer, B.; G{\"{u}}l{\c{c}}ehre, {\c{C}}.; Bahdanau, D.;
  Bougares, F.; Schwenk, H.; and Bengio, Y. 2014.
\newblock Learning Phrase Representations using {RNN} Encoder-Decoder for
  Statistical Machine Translation.
\newblock In Moschitti, A.; Pang, B.; and Daelemans, W., eds.,
  \emph{Proceedings of the 2014 Conference on Empirical Methods in Natural
  Language Processing, {EMNLP} 2014, October 25-29, 2014, Doha, Qatar, {A}
  meeting of SIGDAT, a Special Interest Group of the {ACL}}, 1724--1734. {ACL}.

\bibitem[{Christianos, Papoudakis, and Albrecht(2022)}]{christianos2022pareto}
Christianos, F.; Papoudakis, G.; and Albrecht, S.~V. 2022.
\newblock Pareto Actor-Critic for Equilibrium Selection in Multi-Agent
  Reinforcement Learning.
\newblock arXiv:2209.14344.

\bibitem[{Claus and Boutilier(1998{\natexlab{a}})}]{10.5555/295240.295800}
Claus, C.; and Boutilier, C. 1998{\natexlab{a}}.
\newblock The Dynamics of Reinforcement Learning in Cooperative Multiagent
  Systems.
\newblock In \emph{Proceedings of the Fifteenth National/Tenth Conference on
  Artificial Intelligence/Innovative Applications of Artificial Intelligence},
  AAAI '98/IAAI '98, 746–752. USA: American Association for Artificial
  Intelligence.
\newblock ISBN 0262510987.

\bibitem[{Claus and Boutilier(1998{\natexlab{b}})}]{claus1998dynamics}
Claus, C.; and Boutilier, C. 1998{\natexlab{b}}.
\newblock The dynamics of reinforcement learning in cooperative multiagent
  systems.
\newblock \emph{AAAI/IAAI}, 1998(746-752): 2.

\bibitem[{Ellis et~al.(2022)Ellis, Moalla, Samvelyan, Sun, Mahajan, Foerster,
  and Whiteson}]{ellis2022smacv2}
Ellis, B.; Moalla, S.; Samvelyan, M.; Sun, M.; Mahajan, A.; Foerster, J.~N.;
  and Whiteson, S. 2022.
\newblock SMACv2: An Improved Benchmark for Cooperative Multi-Agent
  Reinforcement Learning.
\newblock arXiv:2212.07489.

\bibitem[{Foerster et~al.(2018)Foerster, Farquhar, Afouras, Nardelli, and
  Whiteson}]{10.5555/3504035.3504398}
Foerster, J.~N.; Farquhar, G.; Afouras, T.; Nardelli, N.; and Whiteson, S.
  2018.
\newblock Counterfactual Multi-Agent Policy Gradients.
\newblock In \emph{Proceedings of the Thirty-Second AAAI Conference on
  Artificial Intelligence and Thirtieth Innovative Applications of Artificial
  Intelligence Conference and Eighth AAAI Symposium on Educational Advances in
  Artificial Intelligence}, AAAI'18/IAAI'18/EAAI'18. AAAI Press.
\newblock ISBN 978-1-57735-800-8.

\bibitem[{Fu et~al.(2022)Fu, Yu, Xu, Yang, and Wu}]{pmlr-v162-fu22d}
Fu, W.; Yu, C.; Xu, Z.; Yang, J.; and Wu, Y. 2022.
\newblock Revisiting Some Common Practices in Cooperative Multi-Agent
  Reinforcement Learning.
\newblock In Chaudhuri, K.; Jegelka, S.; Song, L.; Szepesvari, C.; Niu, G.; and
  Sabato, S., eds., \emph{Proceedings of the 39th International Conference on
  Machine Learning}, volume 162 of \emph{Proceedings of Machine Learning
  Research}, 6863--6877. PMLR.

\bibitem[{Haarnoja et~al.(2019)Haarnoja, Zhou, Hartikainen, Tucker, Ha, Tan,
  Kumar, Zhu, Gupta, Abbeel, and Levine}]{haarnoja2019soft}
Haarnoja, T.; Zhou, A.; Hartikainen, K.; Tucker, G.; Ha, S.; Tan, J.; Kumar,
  V.; Zhu, H.; Gupta, A.; Abbeel, P.; and Levine, S. 2019.
\newblock Soft Actor-Critic Algorithms and Applications.
\newblock arXiv:1812.05905.

\bibitem[{He et~al.(2015)He, Zhang, Ren, and Sun}]{he2015deep}
He, K.; Zhang, X.; Ren, S.; and Sun, J. 2015.
\newblock Deep Residual Learning for Image Recognition.
\newblock arXiv:1512.03385.

\bibitem[{Kingma, Salimans, and Welling(2015)}]{NIPS2015_bc731692}
Kingma, D.~P.; Salimans, T.; and Welling, M. 2015.
\newblock Variational Dropout and the Local Reparameterization Trick.
\newblock In Cortes, C.; Lawrence, N.; Lee, D.; Sugiyama, M.; and Garnett, R.,
  eds., \emph{Advances in Neural Information Processing Systems}, volume~28.
  Curran Associates, Inc.

\bibitem[{Kingma and Welling(2022)}]{kingma2022autoencoding}
Kingma, D.~P.; and Welling, M. 2022.
\newblock Auto-Encoding Variational Bayes.
\newblock arXiv:1312.6114.

\bibitem[{Kuba et~al.(2022)Kuba, Chen, Wen, Wen, Sun, Wang, and
  Yang}]{kuba2022trust}
Kuba, J.~G.; Chen, R.; Wen, M.; Wen, Y.; Sun, F.; Wang, J.; and Yang, Y. 2022.
\newblock Trust Region Policy Optimisation in Multi-Agent Reinforcement
  Learning.
\newblock arXiv:2109.11251.

\bibitem[{Kurach et~al.(2020)Kurach, Raichuk, Stańczyk, Zajac, Bachem,
  Espeholt, Riquelme, Vincent, Michalski, Bousquet, and
  Gelly}]{Kurach_Raichuk_Stańczyk_Zając_Bachem_Espeholt_Riquelme_Vincent_Michalski_Bousquet_Gelly_2020}
Kurach, K.; Raichuk, A.; Stańczyk, P.; Zajac, M.; Bachem, O.; Espeholt, L.;
  Riquelme, C.; Vincent, D.; Michalski, M.; Bousquet, O.; and Gelly, S. 2020.
\newblock Google Research Football: A Novel Reinforcement Learning Environment.
\newblock \emph{Proceedings of the AAAI Conference on Artificial Intelligence},
  34(04): 4501--4510.

\bibitem[{Li et~al.(2023)Li, Liu, Zhang, Wei, Niu, Yang, Liu, and
  Ouyang}]{li2023ace}
Li, C.; Liu, J.; Zhang, Y.; Wei, Y.; Niu, Y.; Yang, Y.; Liu, Y.; and Ouyang, W.
  2023.
\newblock ACE: Cooperative Multi-agent Q-learning with Bidirectional
  Action-Dependency.
\newblock In \emph{Proceedings of the AAAI Conference on Artificial
  Intelligence}.

\bibitem[{Littman(1994)}]{10.5555/3091574.3091594}
Littman, M.~L. 1994.
\newblock Markov Games as a Framework for Multi-Agent Reinforcement Learning.
\newblock In \emph{Proceedings of the Eleventh International Conference on
  International Conference on Machine Learning}, ICML'94, 157–163. San
  Francisco, CA, USA: Morgan Kaufmann Publishers Inc.
\newblock ISBN 1558603352.

\bibitem[{Lowe et~al.(2017)Lowe, Wu, Tamar, Harb, Abbeel, and
  Mordatch}]{10.5555/3295222.3295385}
Lowe, R.; Wu, Y.; Tamar, A.; Harb, J.; Abbeel, P.; and Mordatch, I. 2017.
\newblock Multi-Agent Actor-Critic for Mixed Cooperative-Competitive
  Environments.
\newblock In \emph{Proceedings of the 31st International Conference on Neural
  Information Processing Systems}, NIPS'17, 6382–6393. Red Hook, NY, USA:
  Curran Associates Inc.
\newblock ISBN 9781510860964.

\bibitem[{Ma and Wu(2020)}]{DBLP:conf/atal/MaW20}
Ma, J.; and Wu, F. 2020.
\newblock Feudal Multi-Agent Deep Reinforcement Learning for Traffic Signal
  Control.
\newblock In Seghrouchni, A. E.~F.; Sukthankar, G.; An, B.; and Yorke{-}Smith,
  N., eds., \emph{Proceedings of the 19th International Conference on
  Autonomous Agents and Multiagent Systems, {AAMAS} '20, Auckland, New Zealand,
  May 9-13, 2020}, 816--824. International Foundation for Autonomous Agents and
  Multiagent Systems.

\bibitem[{Müller, Ohström, and Lindenberger(2021)}]{MULLER2021661}
Müller, V.; Ohström, K.-R.~P.; and Lindenberger, U. 2021.
\newblock Interactive brains, social minds: Neural and physiological mechanisms
  of interpersonal action coordination.
\newblock \emph{Neuroscience \& Biobehavioral Reviews}, 128: 661--677.

\bibitem[{Oliehoek and Amato(2016)}]{10.5555/2967142}
Oliehoek, F.~A.; and Amato, C. 2016.
\newblock \emph{A Concise Introduction to Decentralized POMDPs}.
\newblock Springer Publishing Company, Incorporated, 1st edition.
\newblock ISBN 3319289276.

\bibitem[{Peng et~al.(2021)Peng, Rashid, Schroeder~de Witt, Kamienny, Torr,
  B{\"o}hmer, and Whiteson}]{peng2021facmac}
Peng, B.; Rashid, T.; Schroeder~de Witt, C.; Kamienny, P.-A.; Torr, P.;
  B{\"o}hmer, W.; and Whiteson, S. 2021.
\newblock Facmac: Factored multi-agent centralised policy gradients.
\newblock \emph{Advances in Neural Information Processing Systems}, 34:
  12208--12221.

\bibitem[{Rashid et~al.(2020)Rashid, Samvelyan, De~Witt, Farquhar, Foerster,
  and Whiteson}]{10.5555/3455716.3455894}
Rashid, T.; Samvelyan, M.; De~Witt, C.~S.; Farquhar, G.; Foerster, J.; and
  Whiteson, S. 2020.
\newblock Monotonic Value Function Factorisation for Deep Multi-Agent
  Reinforcement Learning.
\newblock \emph{J. Mach. Learn. Res.}, 21(1).

\bibitem[{Ruan et~al.(2022)Ruan, Du, Xiong, Xing, Li, Meng, Zhang, Wang, and
  Xu}]{10.5555/3535850.3535976}
Ruan, J.; Du, Y.; Xiong, X.; Xing, D.; Li, X.; Meng, L.; Zhang, H.; Wang, J.;
  and Xu, B. 2022.
\newblock GCS: Graph-Based Coordination Strategy for Multi-Agent Reinforcement
  Learning.
\newblock In \emph{Proceedings of the 21st International Conference on
  Autonomous Agents and Multiagent Systems}, AAMAS '22, 1128–1136. Richland,
  SC: International Foundation for Autonomous Agents and Multiagent Systems.
\newblock ISBN 9781450392136.

\bibitem[{Schulman et~al.(2018)Schulman, Moritz, Levine, Jordan, and
  Abbeel}]{schulman2018highdimensional}
Schulman, J.; Moritz, P.; Levine, S.; Jordan, M.; and Abbeel, P. 2018.
\newblock High-Dimensional Continuous Control Using Generalized Advantage
  Estimation.
\newblock arXiv:1506.02438.

\bibitem[{Schulman et~al.(2017)Schulman, Wolski, Dhariwal, Radford, and
  Klimov}]{schulman2017proximal}
Schulman, J.; Wolski, F.; Dhariwal, P.; Radford, A.; and Klimov, O. 2017.
\newblock Proximal Policy Optimization Algorithms.
\newblock arXiv:1707.06347.

\bibitem[{Shalev{-}Shwartz, Shammah, and
  Shashua(2016)}]{DBLP:journals/corr/Shalev-ShwartzS16a}
Shalev{-}Shwartz, S.; Shammah, S.; and Shashua, A. 2016.
\newblock Safe, Multi-Agent, Reinforcement Learning for Autonomous Driving.
\newblock \emph{CoRR}, abs/1610.03295.

\bibitem[{Son et~al.(2019)Son, Kim, Kang, Hostallero, and Yi}]{pmlr-v97-son19a}
Son, K.; Kim, D.; Kang, W.~J.; Hostallero, D.~E.; and Yi, Y. 2019.
\newblock {QTRAN}: Learning to Factorize with Transformation for Cooperative
  Multi-Agent Reinforcement Learning.
\newblock In Chaudhuri, K.; and Salakhutdinov, R., eds., \emph{Proceedings of
  the 36th International Conference on Machine Learning}, volume~97 of
  \emph{Proceedings of Machine Learning Research}, 5887--5896. PMLR.

\bibitem[{Sunehag et~al.(2018)Sunehag, Lever, Gruslys, Czarnecki, Zambaldi,
  Jaderberg, Lanctot, Sonnerat, Leibo, Tuyls, and
  Graepel}]{10.5555/3237383.3238080}
Sunehag, P.; Lever, G.; Gruslys, A.; Czarnecki, W.~M.; Zambaldi, V.; Jaderberg,
  M.; Lanctot, M.; Sonnerat, N.; Leibo, J.~Z.; Tuyls, K.; and Graepel, T. 2018.
\newblock Value-Decomposition Networks For Cooperative Multi-Agent Learning
  Based On Team Reward.
\newblock In \emph{Proceedings of the 17th International Conference on
  Autonomous Agents and MultiAgent Systems}, AAMAS '18, 2085–2087. Richland,
  SC: International Foundation for Autonomous Agents and Multiagent Systems.

\bibitem[{Sutton and Barto(2018)}]{sutton2018reinforcement}
Sutton, R.~S.; and Barto, A.~G. 2018.
\newblock \emph{Reinforcement learning: An introduction}.
\newblock MIT press.

\bibitem[{Sutton et~al.(1999)Sutton, McAllester, Singh, and
  Mansour}]{NIPS1999_464d828b}
Sutton, R.~S.; McAllester, D.; Singh, S.; and Mansour, Y. 1999.
\newblock Policy Gradient Methods for Reinforcement Learning with Function
  Approximation.
\newblock In Solla, S.; Leen, T.; and M\"{u}ller, K., eds., \emph{Advances in
  Neural Information Processing Systems}, volume~12. MIT Press.

\bibitem[{Vaswani et~al.(2017)Vaswani, Shazeer, Parmar, Uszkoreit, Jones,
  Gomez, Kaiser, and Polosukhin}]{NIPS2017_3f5ee243}
Vaswani, A.; Shazeer, N.; Parmar, N.; Uszkoreit, J.; Jones, L.; Gomez, A.~N.;
  Kaiser, L.~u.; and Polosukhin, I. 2017.
\newblock Attention is All you Need.
\newblock In Guyon, I.; Luxburg, U.~V.; Bengio, S.; Wallach, H.; Fergus, R.;
  Vishwanathan, S.; and Garnett, R., eds., \emph{Advances in Neural Information
  Processing Systems}, volume~30. Curran Associates, Inc.

\bibitem[{Wan et~al.(2022)Wan, Liu, Chen, Lan, and Zheng}]{pmlr-v162-wan22c}
Wan, L.; Liu, Z.; Chen, X.; Lan, X.; and Zheng, N. 2022.
\newblock Greedy based Value Representation for Optimal Coordination in
  Multi-agent Reinforcement Learning.
\newblock In Chaudhuri, K.; Jegelka, S.; Song, L.; Szepesvari, C.; Niu, G.; and
  Sabato, S., eds., \emph{Proceedings of the 39th International Conference on
  Machine Learning}, volume 162 of \emph{Proceedings of Machine Learning
  Research}, 22512--22535. PMLR.

\bibitem[{Wang et~al.(2021{\natexlab{a}})Wang, Ren, Liu, Yu, and
  Zhang}]{wang2021qplex}
Wang, J.; Ren, Z.; Liu, T.; Yu, Y.; and Zhang, C. 2021{\natexlab{a}}.
\newblock {\{}QPLEX{\}}: Duplex Dueling Multi-Agent Q-Learning.
\newblock In \emph{International Conference on Learning Representations}.

\bibitem[{Wang, Ye, and Lu(2023)}]{wang2023more}
Wang, J.; Ye, D.; and Lu, Z. 2023.
\newblock More Centralized Training, Still Decentralized Execution: Multi-Agent
  Conditional Policy Factorization.
\newblock In \emph{The Eleventh International Conference on Learning
  Representations}.

\bibitem[{Wang et~al.(2023)Wang, Tian, Wan, Wen, Wang, and
  Zhang}]{wang2023order}
Wang, X.; Tian, Z.; Wan, Z.; Wen, Y.; Wang, J.; and Zhang, W. 2023.
\newblock Order Matters: Agent-by-agent Policy Optimization.
\newblock In \emph{The Eleventh International Conference on Learning
  Representations}.

\bibitem[{Wang et~al.(2021{\natexlab{b}})Wang, Han, Wang, Dong, and
  Zhang}]{wang2021dop}
Wang, Y.; Han, B.; Wang, T.; Dong, H.; and Zhang, C. 2021{\natexlab{b}}.
\newblock {\{}DOP{\}}: Off-Policy Multi-Agent Decomposed Policy Gradients.
\newblock In \emph{International Conference on Learning Representations}.

\bibitem[{Wei et~al.(2018)Wei, Wicke, Freelan, and
  Luke}]{DBLP:conf/aaaiss/WeiWFL18}
Wei, E.; Wicke, D.; Freelan, D.; and Luke, S. 2018.
\newblock Multiagent Soft Q-Learning.
\newblock In \emph{2018 {AAAI} Spring Symposia, Stanford University, Palo Alto,
  California, USA, March 26-28, 2018}. {AAAI} Press.

\bibitem[{Wen et~al.(2022)Wen, Kuba, Lin, Zhang, Wen, Wang, and
  Yang}]{NEURIPS2022_69413f87}
Wen, M.; Kuba, J.; Lin, R.; Zhang, W.; Wen, Y.; Wang, J.; and Yang, Y. 2022.
\newblock Multi-Agent Reinforcement Learning is a Sequence Modeling Problem.
\newblock In Koyejo, S.; Mohamed, S.; Agarwal, A.; Belgrave, D.; Cho, K.; and
  Oh, A., eds., \emph{Advances in Neural Information Processing Systems},
  volume~35, 16509--16521. Curran Associates, Inc.

\bibitem[{Yang et~al.(2022)Yang, Dong, Ren, Wang, Wang, and
  Zhang}]{pmlr-v162-yang22a}
Yang, Q.; Dong, W.; Ren, Z.; Wang, J.; Wang, T.; and Zhang, C. 2022.
\newblock Self-Organized Polynomial-Time Coordination Graphs.
\newblock In Chaudhuri, K.; Jegelka, S.; Song, L.; Szepesvari, C.; Niu, G.; and
  Sabato, S., eds., \emph{Proceedings of the 39th International Conference on
  Machine Learning}, volume 162 of \emph{Proceedings of Machine Learning
  Research}, 24963--24979. PMLR.

\bibitem[{Ye et~al.(2023)Ye, Li, Wang, and Zhang}]{ye2023global}
Ye, J.; Li, C.; Wang, J.; and Zhang, C. 2023.
\newblock Towards Global Optimality in Cooperative MARL with the Transformation
  And Distillation Framework.
\newblock arXiv:2207.11143.

\bibitem[{Yu et~al.(2022)Yu, Velu, Vinitsky, Gao, Wang, Bayen, and
  WU}]{NEURIPS2022_9c1535a0}
Yu, C.; Velu, A.; Vinitsky, E.; Gao, J.; Wang, Y.; Bayen, A.; and WU, Y. 2022.
\newblock The Surprising Effectiveness of PPO in Cooperative Multi-Agent Games.
\newblock In Koyejo, S.; Mohamed, S.; Agarwal, A.; Belgrave, D.; Cho, K.; and
  Oh, A., eds., \emph{Advances in Neural Information Processing Systems},
  volume~35, 24611--24624. Curran Associates, Inc.

\bibitem[{Yu et~al.(2023)Yu, Yang, Gao, Chen, Li, Liu, Xiang, Huang, Yang, Wu,
  and Wang}]{10.5555/3545946.3598752}
Yu, C.; Yang, X.; Gao, J.; Chen, J.; Li, Y.; Liu, J.; Xiang, Y.; Huang, R.;
  Yang, H.; Wu, Y.; and Wang, Y. 2023.
\newblock Asynchronous Multi-Agent Reinforcement Learning for Efficient
  Real-Time Multi-Robot Cooperative Exploration.
\newblock In \emph{Proceedings of the 2023 International Conference on
  Autonomous Agents and Multiagent Systems}, AAMAS '23, 1107–1115. Richland,
  SC: International Foundation for Autonomous Agents and Multiagent Systems.
\newblock ISBN 9781450394321.

\bibitem[{Zang et~al.(2023)Zang, He, Li, Fu, Fu, and
  Xing}]{10.5555/3545946.3598674}
Zang, Y.; He, J.; Li, K.; Fu, H.; Fu, Q.; and Xing, J. 2023.
\newblock Sequential Cooperative Multi-Agent Reinforcement Learning.
\newblock In \emph{Proceedings of the 2023 International Conference on
  Autonomous Agents and Multiagent Systems}, AAMAS '23, 485–493. Richland,
  SC: International Foundation for Autonomous Agents and Multiagent Systems.
\newblock ISBN 9781450394321.

\bibitem[{Zhang et~al.(2021)Zhang, Li, Wang, Xie, and Lu}]{pmlr-v139-zhang21m}
Zhang, T.; Li, Y.; Wang, C.; Xie, G.; and Lu, Z. 2021.
\newblock FOP: Factorizing Optimal Joint Policy of Maximum-Entropy Multi-Agent
  Reinforcement Learning.
\newblock In Meila, M.; and Zhang, T., eds., \emph{Proceedings of the 38th
  International Conference on Machine Learning}, volume 139 of
  \emph{Proceedings of Machine Learning Research}, 12491--12500. PMLR.

\end{thebibliography}

\end{document}